\documentclass[12pt]{amsart}

\usepackage{amssymb}
\usepackage[centertags]{amsmath}

\newtheorem%
{thm}{Theorem}[section]
\newtheorem%
{proposition}[thm]{Proposition}
\newtheorem%
{lemma}[thm]{Lemma}
\newtheorem%
{definition}[thm]{Definition}
\newtheorem%
{corollary}[thm]{Corollary}
\newtheorem%
{conjecture}[thm]{Conjecture} \theoremstyle{definition}

\theoremstyle{remark}
\newtheorem{rem}{Remark}[section]

\newcommand{\dontprint}[1]{\relax}

\usepackage[all]{xy}

\title[Variational tricomplex of a local  gauge system]
{Variational tricomplex of a local  gauge system, \\Lagrange structure and weak Poisson bracket}
\author{A. A. Sharapov}
\address{Physics Faculty, Tomsk State University, Tomsk 634050, Russia}
\email{sharapov@phys.tsu.ru}

\begin{document}
\maketitle
\begin{abstract}
We introduce the concept of a variational tricomplex, which is applicable  both to variational and non-variational gauge systems. Assigning this tricomplex with an appropriate symplectic structure and a Cauchy foliation, we establish a general correspondence between the Lagrangian and Hamiltonian pictures of one and the same (not necessarily variational)  dynamics. In practical terms, this correspondence allows one to construct the generating functional of weak Poisson structure starting from that of Lagrange structure. As a byproduct, a covariant procedure is proposed for deriving the classical BRST charge of the BFV formalism by a given BV master action. The general approach is illustrated by the examples of Maxwell's electrodynamics and chiral bosons in two dimensions.
\end{abstract}
\section{Introduction}

The Lagrangian and Hamiltonian formalisms provide two most popular
approaches to classical dynamics. They also serve as departing
points for the procedures of path-integral and canonical
quantizations.  For non-singular Lagrangian theories the equivalence
between the two approaches is established by the Legendre
transformation. The treatment of singular Lagrangians  appears to be
more tricky: besides the Legendre transform, it involves the
Dirac-Bergmann algorithm and leads to the so-called constrained
Hamiltonian dynamics \cite{HT}. In spite of technical differences,
the Lagrangian and Hamiltonian formalisms have one point in common:
either   assumes the classical equations  of motion to come from the
least action principle.  The variational nature of classical
dynamics is thus at the heart of both the formalisms. Promoted at
the quantum level, this feature has been embodied in a wide-spread
belief  that the variational formulation of classical dynamics is
``a must''  prerequisite for the existence of a consistent
quantization. The actual situation, however, is much more
interesting.

In \cite{LS0}, \cite{KazLS}, it was shown that both the Lagrangian and Hamiltonian pictures of dynamics admit nontrivial extensions beyond the scope  of variational principles. The key elements of these  extensions are, respectively, the notions of a Lagrange structure and a weak Poisson structure. Added to the classical equations of motion, these structures make possible a fully consistent quantization of the classical theory along the lines of path-integral or deformation quantization depending on which picture of dynamics, Lagrangian or Hamiltonian, is considered. They also allow one to establish a general correspondence between the conservation lows and symmetries, providing thus a generalization of the seminal Nother's theorem \cite{KLS1}. It is significant that the existence of the aforementioned structures appears to be less restrictive condition for the classical dynamics than the existence of variational principle.   Furthermore, one and the same equations of motion may have a variety of compatible Lagrange or weak Poisson structures leading, in general, to inequivalent quantum theories. Similar to the usual BRST theory of variational gauge systems \cite{HT}, either of the two structures admits a compact formulation in terms of a generating functional and a master equation on the ghost-extended configuration or  phase space of the theory.

The aim of  this paper is to establish  a direct correspondence  between the Lagrangian and Hamiltonian pictures of (non-)variational dynamics at the level of the generating functionals. Establishing of such a correspondence  is a matter of principle; its existence is just as fundamental as the equivalence of the conventional Lagrangian and Hamiltonian formalisms. The present paper can be viewed as a continuation of our previous work \cite{Sh}, where a relation between the Lagrange structure and the weak Poisson bracket was established through the Peierls bracket construction. It should be noted that unlike the Poisson bracket, the Peierls bracket is essentially non-local and becomes local only  in the equal-time limit. This non-locality may be regarded as an unnecessary complication when one is only interested  in deriving the equal-time Poisson bracket and its subsequent deformation quantization. In the next sections, we propose a general construction which is fully local and allows one to define the generating functional of  weak Poisson structure by the corresponding functional for the Lagrange structure. Algebraically, it  links the $S_\infty$- and $P_\infty$-algebras underlying the gauge system. Central to our approach is the concept of a variational bicomplex \cite{Anderson}, which we extend to a tricomplex by adding the classical BRST differential. This allows us to replace the usual calculus of variations by a  more handy calculus of exterior differential forms on jet bundles.

In quite a similar context the variational tricomplex for gauge systems was first introduced in \cite{BH} as the Koszul-Tate resolution of the usual variational bicomplex for partial differential equations.  Using this tricomplex the authors of \cite{BH} were able to relate various Lie algebras associated with the symmetries and conservation laws of a variational gauge system. Our tricomplex is similar in nature but involves the full BRST differential, and not its Koszul-Tate part. Besides,  we do not restrict our consideration to the case of variational theories.

It turns out that the approach we develop below is useful even in
the variational situation in determining a correspondence between
the Batalin-Vilkovisky (BV) formalism for Lagrangian gauge systems
\cite{BV1}, \cite{BV2} and its Hamiltonian counterpart known as the
Batalin-Fradkin-Vilkovisky (BFV) formalism \cite{BFV1}, \cite{BFV2},
\cite{BFV3}. Usually, these are developed in parallel starting,
respectively, from the classical action or the first class
constraints.  In either case one applies the homological
perturbation theory (hpt) to obtain the master action or the
classical BRST charge  at the output \cite{HT}. As already
mentioned,  the relation between both the pictures of dynamics is
established  through the Dirac-Bergmann (DB) algorithm, which allows
one to generate the complete set of the first class constraints by
the classical action. All these can be displayed diagrammatically as
follows:
$$
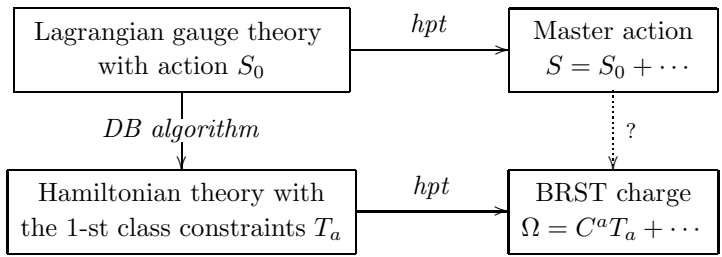
\xymatrix@R=10pt{*+[F]\txt{\Small \;Lagrangian gauge theory\;\;\\\Small with action $S_0$} \ar[rr]^-{\txt{\Small \textit{hpt}}}\ar[dd]|{\txt{\Small \textit{DB algorithm}}} &&*+[F]\txt{\Small Master action \\\Small \;\; $S=S_0+\cdots\;\;$}\ar@{.>}[dd]^{?}\\&&\\
*+[F]\txt{\Small Hamiltonian theory with\\\Small the $1$-st class constraints $T_a$} \ar[rr]^-{\txt{\Small \textit{hpt}}}&&*+[F]\txt{\Small BRST charge \\\Small $\Omega=C^aT_a+\cdots$} }
$$

Looking at this picture it is natural to ask about the dotted arrow making the diagram commute. The arrow symbolizes a hypothetical  map or construction relating the BV formalism to the BFV formalism at the level of generating functionals. As we show below such a map does really exist: By making use of the variational tricomplex, we propose a direct construction of the classical BRST charge  by the BV master action. The construction is explicitly covariant (even though we pass to the Hamiltonian picture)  and generates the full spectrum of the BFV ghosts  immediately from that of the BV theory. We also derive the Poisson bracket on the extended phase space of the theory, with respect  to which the classical BRST charge obeys the master equation. Our definition of the Poisson structure is similar in spirit to that presented in Ref.\cite{Dickey} if not identical in two respects.
 For one thing, we define the Poisson algebra of Hamiltonian forms off-shell; for another,  the definition of the  Hamiltonian forms essentially involves the choice of a causal structure on the underlying space-time manifold.  As  a result, we arrive at a rich Poisson algebra of Hamiltonian forms involving not just the first integrals of motion (cf. \cite[Sec.19.7]{Dickey}).

The paper is organized as follows. In the next section, we formalize the notion of a gauge system in terms of a (foliated) variational bicomplex endowed with a BRST differential and a compatible  presymplectic structure. Here we also define the notion of a descendent gauge system, which is basic to our subsequent considerations.  A covariant relationship between the BV and BFV formalisms is established and illustrated in Sec.3. In Sec.4, it is extended to non-variational gauge systems. More precisely, we show that under certain assumptions each Lagrange structure gives rise to a weak Poisson structure. Appendix A contains some basic facts concerning the geometry of jet bundles and the variational bicomplex.

\section{Variational tricomplex of a local gauge system}

In modern language the classical fields are just the sections of a locally trivial, fiber bundle $\pi : E\rightarrow M$ over an $n$-dimensional space-time manifold $M$. The typical fiber $F$ of $E$ is called the \textit{target space of fields}. In case the bundle is trivial, i.e., $E=M\times F$, the fields are merely  the smooth mappings from $M$ to $F$. For the sake of simplicity, we restrict ourselves to fields associated with vector bundles. In this case the space of fields $\Gamma (E)$ has the structure of a real vector space.

Bearing in mind the gauge theories together with their ghost extension as well as the field theories with fermions, we assume $\pi: E\rightarrow M$ to be a $\mathbb{Z}$-graded supervector bundle over the ordinary (non-graded) smooth manifold $M$. The Grassmann parity and the $\mathbb{Z}$-grading of a homogeneous geometrical object $A$ will be denoted by  $\widetilde{A}$ and $\deg A$, respectively. It should be emphasized that in the presence of fermionic fields there is no natural correlation between the Grassmann parity and the $\mathbb{Z}\,$-grading and it is the Grassmann parity which is responsible for the sign rule. Since  throughout the paper we work exclusively in the category of $ \mathbb{Z}\,$-graded supermanifolds, we omit the boring prefixes ``super'' and ``graded'' whenever possible;  smooth  manifolds, vector bundles, commutators, etc., are understood in the graded sense. For a quick introduction to the graded differential geometry we refer the reader to \cite{Vor1}-\cite{CatSch}.

In the local field theory, the dynamics of fields are governed by partial differential equations.
The best way to account for the local structure of fields is to introduce the  variational bicomplex $\Lambda^{\ast,\ast}(J^\infty E; d, \delta)$ on the infinite jet bundle $J^\infty E$ associated with the vector bundle $\pi: E\rightarrow M$. The free variational bicomplex represents a natural kinematical basis for formulating local field theories\footnote{A brief account of this concept is given in Appendix \ref{A}, where we also explain our notation.}.  In order to specify dynamics two more geometrical ingredients are needed. These are the classical BRST differential and the BRST-invariant (pre)symplectic structure on $J^\infty E$. Let us give the corresponding definitions.

\subsection{Presymplectic structure}\label{PrSt} By a \textit{presymplectic} $(2,m)$-form  on $J^\infty E$ we understand an
element $\omega\in \widetilde{\Lambda}^{2,m}(J^\infty E)$ satisfying
\footnote{By abuse of notation, we denote by $\omega$ an element of the quotient space $\widetilde{\Lambda}^{2,m}=\Lambda^{2,m}/d\Lambda^{2,m-1}$ and its particular representative   in $\Lambda^{2,m}$. The sign $\simeq$ means equality modulo $d\Lambda^{\ast,\ast}$.}
\begin{equation}\label{dom}
\delta \omega \simeq 0\,.
\end{equation}
The form $\omega$ is assumed to be homogeneous, so that we can speak of an odd or even  presymplectic structure of definite $\mathbb{Z}$-degree. Triviality of the relative $\delta$-cohomology in positive vertical degree (Proposition \ref{Prop-A}) implies that any presymplectic $(2,m)$-form is exact, namely, there exists a homogeneous  $(1,m)$-form $\theta$ such that $\omega\simeq\delta\theta$. The form $\theta$ is called the \textit{presymplectic potential} for $\omega$.  Clearly, the presymplectic potential is not unique. If $\theta_0$ is one of the presymplectic potentials for $\omega$, then setting   $\omega_0=\delta \theta_0$ we get
$$
\delta \omega_0=0\,,\qquad \omega_0\simeq \omega\,.
$$
In other words,  any presymplectic form has a  $\delta$-closed representative.

Denote by  $\ker\omega$ the space of all evolutionary vector fields $X$ on $J^\infty E$ that fulfill the relation
$$
i_X\omega \simeq 0\,.
$$
A presymplectic form $\omega$ is called non-degenerate if $\ker \omega=0$, in which case we refer to it as a \textit{symplectic form}.

An evolutionary vector field $X$ is called \textit{Hamiltonian} with respect to $\omega$ if it preserves the presymplectic form, that is,
\begin{equation}\label{Xom}
L_X\omega\simeq 0\,.
\end{equation}
Obviously, the Hamiltonian vector fields form a subalgebra in the Lie algebra of all evolutionary vector fields. We denote this subalgebra by $\frak{X}_\omega(J^\infty E)$. In view (\ref{dom}), Eq. (\ref{Xom}) is equivalent to
$$
\delta i_X\omega\simeq 0\,.
$$
Again, because of the triviality of the relative $\delta$-cohomology, we conclude that
\begin{equation}\label{HVF}
i_X\omega \simeq \delta H
\end{equation}
for some $H\in \widetilde{\Lambda}^{0,m}(J^\infty E)$. We refer to $H$ as a \textit{Hamiltonian form} (or   \textit{Hamiltonian}) associated with $X$.  Sometimes, to indicate the relationship between the Hamiltonian vector fields and forms, we will write $X_H$ for $X$.  In general, the relationship is  far from being one-to-one. On the one hand, we are free to add to $X$ any vector field from $\ker \omega$ keeping the Hamiltonian $H$ intact, and on the other we can add to $H$ any element of $\Lambda^m(M)$ whenever $\widetilde{H}=0$ and $\deg H=0$.

The space ${\Lambda}_\omega^{0,m}(J^\infty E)$ of all Hamiltonian $m$-forms can be endowed with the structure of a Lie algebra. The corresponding Lie bracket is defined as follows: If $X_A$ and $X_B$ are two Hamiltonian vector fields associated with the Hamiltonian forms $A$ and $B$, then
\begin{equation}\label{PB}
\{A,B\}=(-1)^{\widetilde{X}_A}i_{X_A}i_{X_B}\omega\,.
\end{equation}
The next proposition shows that the bracket is well defined and possesses  all the required properties.

\begin{proposition}\label{2.1}
The bracket (\ref{PB}) is bilinear over reals, maps the Hamiltonian forms to Hamiltonian ones, enjoys the symmetry property
\begin{equation}\label{sym}
\{A,B\}\simeq -(-1)^{(\widetilde{A}+\widetilde{\omega})(\widetilde{B}+\widetilde{\omega})}\{B,A\}\,,
\end{equation}
and obeys  the Jacobi identity
\begin{equation}\label{jac}
\{C,\{A,B\}\}\simeq \{\{C,A\},B\}+(-1)^{(\widetilde{C}+\widetilde{\omega})(\widetilde{A}+\widetilde{\omega})}\{A,\{C,B\}\}\,.
\end{equation}
\end{proposition}

\begin{proof}
Bilinearity is obvious. It is also clear that (\ref{PB}) does not
depend on the choice of the Hamiltonian vector fields $X_A$ and
$X_B$.

Now the symmetry property follows from the chain of relations
$$
\begin{array}{rcl}
\{A,B\}&\simeq& (-1)^{\widetilde{X}_A}i_{X_A}i_{X_B}\omega=(-1)^{\widetilde{X}_A+(\widetilde{X}_A+1)(\widetilde{X}_B+1)}i_{X_B}i_{X_A}\omega\\[2mm]
&\simeq& (-1)^{\widetilde{X}_A+\widetilde{X}_B+(\widetilde{X}_A+1)(\widetilde{X}_B+1)}\{B,A\}\\[2mm]
&=&-(-1)^{(\widetilde{A}+\widetilde{\omega})(\widetilde{B}+\widetilde{\omega})}\{B,A\}\,.
\end{array}
$$
Here we used the equality $\widetilde{X}_H=\widetilde{H}+\widetilde{\omega}$, which readily follows from the definition (\ref{HVF}).

In order to prove the remaining assertions consider the following equalities:
$$
\begin{array}{rl}
0&\simeq i_{X_A}i_{X_B}\delta \omega = i_{X_A}L_{X_B} \omega -(-1)^{\widetilde{X}_B}i_{X_A}\delta i_{X_B}\omega\\[2mm]
&=(-1)^{(\widetilde{X}_A+1)\widetilde{X}_B} (L_{X_B}i_{X_A}\omega-i_{[X_B,X_A]}\omega) -(-1)^{\widetilde{X}_B}i_{X_A}\delta i_{X_B}\omega\\[2mm]
&\simeq (-1)^{(\widetilde{X}_A+1)\widetilde{X}_B} (L_{X_B}\delta A-i_{[X_B,X_A]}\omega) - (-1)^{\widetilde{X}_B}i_{X_A}\delta^2 B\\[2mm]
&\simeq (-1)^{(\widetilde{X}_A+1)\widetilde{X}_B}( (-1)^{\widetilde{X}_B}\delta L_{X_B} A-i_{[X_B,X_A]}\omega)\\[2mm]
&\simeq (-1)^{(\widetilde{X}_A+1)\widetilde{X}_B}(\delta \{B,A\}-i_{[X_B,X_A]}\omega)\,.
\end{array}
$$
 We see that the form $\{A,B\}$ is Hamiltonian and  corresponds to the Hamiltonian vector field $[X_A,X_B]$. At the same time this proves the Jacobi identity
$$
\begin{array}{rcl}
\{C,\{A,B\}\}&=&(-1)^{\widetilde{X}_C}L_{X_C}\{A,B\}=(-1)^{\widetilde{X}_C+\widetilde{X}_A}L_{X_C}i_{X_A}i_{X_B}\omega\\[2mm]
&\simeq&(-1)^{\widetilde{X}_C+\widetilde{X}_A}i_{[X_C,X_A]}i_{X_B}\omega +(-1)^{\widetilde{X}_C+\widetilde{X}_A+\widetilde{X}_C(\widetilde{X}_A+1)}i_{X_A}L_{X_C}i_{X_B}\omega\\[2mm]
&\simeq& \{\{C,A\},B\}+(-1)^{\widetilde{X}_C\widetilde{X}_A+\widetilde{X}_A}i_{X_A}i_{[X_C, X_B]}\omega \\[2mm] &+&(-1)^{\widetilde{X}_C+\widetilde{X}_A+\widetilde{X}_C(\widetilde{X}_A+1)+\widetilde{X}_C(\widetilde{X}_B+1)}i_{X_A}i_{X_B}L_{X_C}\omega\\[2mm]
&\simeq& \{\{C,A\},B\}+(-1)^{\widetilde{X}_A\widetilde{X}_C}\{A,\{C,B\}\}\,.
\end{array}
$$

\end{proof}

\subsection{Classical BRST differential}\label{brst}
An odd evolutionary  vector field $Q$ on $J^\infty E$ is called \textit{homological}\footnote{Some authors prefer the term \textit{cohomological} \cite{CatSch}.} if
\begin{equation}\label{QQ}
[Q,Q]=2Q^2=0\,, \qquad \deg\,Q=1\,.
\end{equation}
The Lie derivative along the homological vector field $Q$ will be denoted by $\delta_Q$. It follows from the definition that $\delta_Q^2=0$. Hence, $\delta_Q$ is a differential of the algebra $\Lambda^{\ast,\ast}(J^\infty E)$ increasing the $\mathbb{Z}$-degree by 1. Moreover, the operator $\delta_Q$ anticommutes with the coboundary operators $d$ and $\delta$:
$$
\delta_Q d+d\delta_Q=0\,,\qquad \delta_Q \delta+\delta\delta_Q=0\,.
$$
This allows us to speak of the tricomplex $\Lambda^{\ast,\ast,\ast}(J^\infty E; d, \delta, \delta_Q)$, where
$$
\delta_Q: \Lambda^{p,q,r}(J^\infty E)\rightarrow \Lambda^{p,q,r+1}(J^\infty E)\,.
$$

 In the physical literature the homological vector field $Q$ is known as the \textit{classical BRST differential}  and the $\mathbb{Z}$-grading is called the \textit{ghost number}. These are the two main ingredients of all modern approaches to the covariant quantization of gauge theories.  In the BV formalism, for example, the BRST differential carries all the information about equations of motions, their  gauge symmetries and identities, and the space of physical observables is naturally identified with the group  $H^{0,{n},0}(J^\infty E; \delta_Q/d) $ of ``$\delta_Q$ modulo $d$'' cohomology in ghost number zero. For general non-Lagrangian gauge theories the classical BRST differential was systematically defined in \cite{LS0}, \cite{KazLS}.

The equations of motion of a gauge theory can be recovered by considering the zero locus of the homological vector field $Q$. In terms of adapted coordinates $(x^i, \phi^a_I)$ on $J^\infty E$ the vector  field $Q$, being evolutionary, assumes the form
$$
Q=\partial_I Q^a\frac{\partial}{\partial \phi_I^a}\,.
$$
Then there exists an integer $l$ such that the equations
$$
\partial_I Q^a=0\,,\qquad |I|=k\,,
$$
define a submanifold $\Sigma^k\subset J^{l+k}E$. The standard regularity condition imposed usually on $Q$ is that $\Sigma^0$ is a smooth, closed subbundle of $J^l E$, and  $\Sigma^{k+1}$ fibers over $\Sigma^k$ for each $k$. This gives the infinite sequence of projections

$$
\xymatrix{\cdots\ar[r]& \Sigma^{l+3}\ar[r]&\Sigma^{l+2}\ar[r]&\Sigma^{l+1}\ar[r]&\Sigma^l}\rightarrow M\,,
$$
which enables us to define the zero locus of $Q$ as the inverse limit
 $$
 \Sigma^\infty =\lim_{\longleftarrow}\Sigma^k\,.
 $$
In physics, the submanifold $\Sigma^\infty\subset J^\infty E$ is usually referred to as  the \textit{shell}. The terminology is justified by the fact that the classical field equations as well as their differential consequences can be written as \footnote{In the conventional BRST theory of variational gauge systems, the relationship between the zero locus of the classical BRST differential and solutions to the classical equations of motion was studied in \cite{GST}.}
$$
(j^{\infty}\phi)^\ast (\partial_I Q^a)=0\,.
$$
In other words, the field $\phi\in \Gamma(E)$ satisfies the classical equations of motion  iff $j^\infty \phi \in \Sigma^\infty$.

Unlike $\Sigma^k$, the shell $\Sigma^\infty$ is invariant under the action of $Q$ as one can readily see from (\ref{QQ}). This makes possible to pull the ``free'' variational tricomplex $\Lambda^{\ast,\ast,\ast}(J^\infty E; d,\delta,\delta_Q)$ back to $\Sigma^\infty$ and so define the \textit{on-shell tricomplex} $\Lambda^{\ast,\ast,\ast}(\Sigma^\infty; d,\delta,\delta_Q)$.
 The latter  is not generally $d$-exact even locally and this gives rise to various interesting cohomology groups associated with gauge dynamics. For example, the elements of the group $H^{0,n-1,0}(\Sigma^\infty; d)$ are naturally identified with the nontrivial conservation laws. Interpretation of some other groups can be found in \cite{BBH}, \cite{KLS2}.  We will not further expand on the properties of the on-shell tricomplex as in our subsequent considerations we  mostly deal with the free variational tricomplex.

\subsection{$Q$-invariant presymplectic structure and its descendants}
By a \textit{gauge system} on $J^\infty E$ we will mean a pair $(Q, \omega)$ consisting of a homological vector field $Q$ and a $Q$-invariant presymplectic $(2,m)$-form $\omega$. In other words, the vector field $Q$ is supposed to be Hamiltonian with respect to $\omega$, so that $\delta_Q\omega\simeq 0$. The last relation implies the existence  of  forms $\omega_1$, $H$, and $\theta_1$ such that
\begin{equation}\label{des}
\delta_Q \omega=d\omega_1\,, \qquad   i_Q\omega =\delta H +d\theta_1\,.
\end{equation}
 As was mentioned in Sec.\ref{PrSt}, we can always assume that $\omega =\delta\theta$ for some presymplectic potential $\theta$, so that  $\delta\omega=0$. Then applying $\delta$ to the second equality in (\ref{des}) and using the first one, we find $d(\omega_1-\delta\theta_1)=0$. On account of the exactness of the variational bicomplex, the last relation is equivalent to
$$
\omega_1\simeq \delta\theta_1\,.
$$
Thus, $\omega_1$ is a presymplectic $(2,m-1)$-form on $J^\infty E$ coming from the presymplectic potential $\theta_1$. Furthermore, the form $\omega_1$ is $Q$-invariant as one can easily see by applying $\delta_Q$ to the first equality in (\ref{des}) and using once again the fact of exactness of the variational bicomplex. Let $H_1$ denote the Hamiltonian for $Q$ with respect to $\omega_1$, i.e.,
$$
i_Q\omega_1\simeq \delta H_1\,, \qquad H_1\in \widetilde{\Lambda}^{0,m-1}(J^\infty E)\,.
$$
It follows from the definitions that
$$
\begin{array}{ll}
 \widetilde{\omega}_1=\widetilde{\omega}+1 \quad (\mathrm{mod}\; 2)\,,& \qquad \deg \omega_1=\deg \omega+1\,,\\[2mm]
\widetilde{H}_1=\widetilde{H}+1 \quad (\mathrm{mod}\; 2)\,,&\qquad \deg H_1=\deg H+1\,.
\end{array}
$$

Given the pair $(Q,\omega)$, we call $\omega_1$ the \textit{descendent presymplectic structure} on $J^\infty E$ and refer to $(Q,\omega_1)$ as the \textit{descendent gauge system}.

The next proposition provides an alternative definition for the descendent Hamiltonian of the homological vector field.
\begin{proposition}\label{p2}
Let $\omega$ be a $\delta$-closed representative of a presymplectic $(2,m)$-form on $J^\infty E$ and $\mathrm{deg} H_1\neq 0$, then
\begin{equation}\label{HH}
dH_1=-\frac12\{H,H\}\,.
\end{equation}
\end{proposition}
\begin{proof}
$$
\begin{array}{rcl}
0&=&i^2_Q\delta\omega=i_Q(\delta_Q+\delta i_Q)\omega=(\delta_Q+i_Q\delta)i_Q\omega\\[2mm]
&=&(2i_Q\delta-\delta i_Q)(\delta H+d\theta)=2i_Q\delta d\theta-\delta i_Q\delta H-\delta i_Qd\theta\\[2mm]
&=&-2di_Q\delta\theta-\delta(\delta_QH+di_Q\theta)=-2di_Q\omega_1-\delta i^2_Q\omega\\[2mm]
&=&-2d\delta H_1-\delta i^2_Q\omega =\delta (2dH_1+\{H,H\})\,.
\end{array}
$$
Hence $(2dH_1+\{H,H\})\in \Lambda^m(M)$. Since  $\deg
(2dH_1+\{H,H\})=\mathrm{deg} H_1\neq 0$, while $\Lambda^m(M)$
belongs to $\mathbb{Z}$-degree zero, we conclude that
$2dH_1+\{H,H\}=0$.
\end{proof}

\begin{corollary}\label{c1}
$H$ is a Maurer-Cartan element of the Lie algebra $\Lambda^{0,m}_\omega(J^\infty E)$, that is, $$\{H,H\}\simeq 0\,.$$
\end{corollary}
Since $H$ and $\omega$ have opposite Grassmann parities, the last equality is not a trivial consequence of the symmetry property (\ref{sym}).

\begin{corollary}\label{CL}
The Hamiltonian form $H_1$ is $d$-closed on-shell. In particular, for $m=n$ it defines a conservation law.
\end{corollary}

Indeed, writing (\ref{HH}) in the form
$$
-\frac12 i^2_Q\omega =dH_1\,,
$$
we see that the l.h.s. vanishes on the zero locus of $Q$. Therefore $H_1$ represents an element of the cohomology group $H^{0, m-1, r}(\Sigma^\infty; d)$ with $r=\deg H_1$. As we have mentioned in Sec.\ref{brst}, the group $H^{0, n-1, r}(\Sigma^\infty; d)$ describes the space of nontrivial conservation laws of the gauge system.

\begin{proposition}\label{p5}
Suppose that the $Q$-invariant presymplectic form $\omega$ of top horizontal degree has the structure
\begin{equation}\label{ff}
\omega=P_{ab}\wedge\delta\phi^a\wedge\delta\phi^b\,,\qquad P_{ab}\in \Lambda^{0,n}(J^\infty E)\,,
\end{equation}
and $H$ is the Hamiltonian of $Q$ with respect to $\omega$. Then the presymplectic potential for the descendent presymplectic (2,n-1)-form  $\omega_1\simeq \delta\theta_1$ is defined  by the equation
\begin{equation}\label{ht1}
\delta H=\delta\phi^a\wedge \frac{\delta H}{\delta\phi^a}-d\theta_1 \,.
\end{equation}
\end{proposition}

\begin{proof}
According to (\ref{des}) the variation of $H$ is given by
\begin{equation}\label{hqt}
\delta H=i_Q\omega-d\theta_1\,.
\end{equation}
Because of the special structure of the presymplectic form (\ref{ff}),   $i_Q\omega$ is a source form. Then in virtue of Proposition \ref{SF} we have
$$
i_Q\omega =\delta\phi^a\wedge \frac{\delta H}{\delta\phi^a}\,.
$$
Combining the last relation with (\ref{hqt}), we get (\ref{ht1}).
\end{proof}

The above construction of the descendent gauge system $(Q,\omega_1)$ can be iterated producing a sequence of gauge systems $(Q, \omega_k)$, where the $k$-th presymlectic form $\omega_k\in {\Lambda}^{2,m-k}(J^\infty E)$ is the descendant of the previous form $\omega_{k-1}$. The minimal $k$ for which $\omega_k \simeq 0$ gives a numerical  invariant of the original gauge system $(Q,\omega)$.

\subsection{Foliated variational tricomplex}
 In previous sections, we have defined the Lie algebra of Hamiltonian forms $\Lambda_\omega^{0,m}(J^\infty E)$  associated with a presymplectic structure $\omega\in \widetilde{ \Lambda}^{2,m}(J^\infty E)$. For gauge systems, this algebra is certainly nonempty as it contains e.g. the Hamiltonian of the classical BRST differential $Q$. It should be noted, however, that in most field-theoretical applications the algebra $\Lambda_\omega^{0,m}(J^\infty E)$ appears too scanty to accommodate all quantities of physical interest whenever $m<n$. On the other hand, the physical quantities are given not by the Hamiltonian $m$-forms \textit{per se}  but their integrals over $m$-chains in $M$. In order to enrich  the algebra of Hamiltonian $m$-forms  one can try to restrict the set of admissible $m$-chains of $M$. By duality, this must extend the space of admissible $m$-cochains, i.e., Hamiltonian $m$-forms.  The physical motivation  for such a restriction comes basically from the relationship between the Lagrangian and Hamiltonian formalisms in field theory. In the Lagrangian picture, the physical quantities are described by local functionals of fields, i.e., integrals over the $n$-dimensional space-time manifold $M$. Passing to the Hamiltonian formalism one has to split the original space-time $M$ into space and time by choosing a global time function. The $n$-dimensional manifold $M$ is then foliated by $(n-1)$-dimensional time slices. On the one hand, these slices play the role of the Cauchy hypersurfaces for the field equations, and on the other hand they represent $(n-1)$-chains of $M$ over which the Hamiltonian $(n-1)$-forms are to be integrated to produce physical quantities  in the Hamiltonian picture. Thus, upon choosing a causal structure on $M$, the admissible $(n-1)$-chains are not arbitrary, rather they constitute a one-parameter family of hypersurfaces in $M$. The relationship between the  Lagrangian and Hamiltonian formulations  of gauge dynamics will be discussed more fully in Sec.\ref{BV-BFV}.

 Proceeding from the above line of reasoning, we can now formulate the general definition of a foliated variational tricomplex.  Consider an infinite jet bundle $\pi_\infty : J^\infty E\rightarrow M$ whose base $M$ is equipped with the structure of a smooth $m$-dimensional foliation $\mathcal{F}(M)$. We assume  the annihilating ideal of forms $I(\mathcal{F})\subset \Lambda^\ast(M)$ to be algebraically generated by a set of linearly independent 1-forms $\eta^A$, $A=1,\ldots, n-m$, such that $\eta^A|_C=0$ for each leaf $C\in \mathcal{F}(M)$. By Frobenius theorem $d I\subset I$ or, what is the same,
 $$
 d\eta^A=\eta^A_B\wedge \eta^B
 $$
for some $\eta^A_B\in \Lambda^1(M)$. Thus, $\mathcal{I}(\mathcal{F})$ is a differential ideal of $\Lambda^\ast(M)$. The quotient $\Lambda^\ast(\mathcal{F})=\Lambda^\ast(M)/I(\mathcal{F})$ is known as the algebra of differential forms along $\mathcal{F}$.

The  ideal $I(\mathcal{F}) \subset \Lambda^\ast(M)$ can be lifted to an ideal $\mathcal{I}(\mathcal{F})$ in $\Lambda^{\ast,\ast}(J^\infty E)$.  As a linear space, $\mathcal{I}(\mathcal{F})$ is generated by the wedge products
$$
\alpha\wedge \beta\,,\qquad \alpha \in \Lambda^{\ast,\ast}(J^\infty E)\,, \qquad \beta \in \pi^\ast_\infty I(\mathcal{F})\,.
$$
 By definition we set $\Lambda^{\ast,\ast}_\mathcal{F}(J^\infty E)=\Lambda^{\ast,\ast}(J^\infty E)/\mathcal{I}(\mathcal{F})$. The space $\mathcal{I}(\mathcal{F})$, being obviously  invariant under the action of all three differentials $d$, $\delta$ and $\delta_Q$,  defines a subcomplex  of the variational tricomplex $\Lambda^{\ast,\ast,\ast}(J^\infty E; d,\delta, \delta_Q)$. By a \textit{foliated tricomplex } we mean the quotient complex
$$
\Lambda^{\ast,\ast,\ast}_{\mathcal{F}}(J^\infty E; d,\delta, \delta_Q)=\Lambda^{\ast,\ast,\ast}(J^\infty E; d,\delta, \delta_Q)/\mathcal{I}(\mathcal{F})\,.
$$

Given a vertical vector field $X$, we can define the operator
$$
\Delta_X=i_Xd+(-1)^{\tilde{X}}di_X\,,
$$
$$
\Delta_X: \Lambda^{p,q}(J^\infty E)\rightarrow \Lambda^{p-1,q+1}(J^\infty E)\,.
$$
A vertical vector field $X$ will be called \textit{$\mathcal{F}$-evolutionary} if $\mathrm{Im}\, \Delta_X \in \mathcal{I}(\mathcal{F})$. In other words, for any $\mathcal{F}$-evolutionary vector field $X$ the operators $i_X$ and $d$ (anti)commute modulo $\mathcal{I}(\mathcal{F})$.
The space of all $\mathcal{F}$-evolutionary vector fields will be denoted by ${\frak X}_\mathcal{F}(J^\infty E)$.
Note that any evolutionary vector field is $\mathcal{F}$-evolutionary, so that ${\frak X}_{ev}(J^\infty E)\subset {\frak X}_\mathcal{F}(J^\infty E)$.

By analogy with $\widetilde{\Lambda}^{\ast,\ast,\ast}(J^\infty E; d,\delta,\delta_Q)$ we  define the quotient tricomplex
\begin{equation}\label{quot}
\widetilde{\Lambda}_\mathcal{F}^{\ast,\ast,\ast}(J^\infty E; d,\delta,\delta_Q)
={\Lambda}_\mathcal{F}^{\ast,\ast,\ast}(J^\infty E; d,\delta,\delta_Q)/
d{\Lambda}_\mathcal{F}^{\ast,\ast,\ast}(J^\infty E; d,\delta,\delta_Q)\,.
\end{equation}
The elements of the latter complex are represented by the differential forms on $J^\infty E$ considered modulo $d$-exact forms and forms belonging to $\mathcal{I}(\mathcal{F})$. For notational simplicity we will denote by the same letter an element of the quotient (\ref{quot}) and its representative in  $\Lambda^{\ast,\ast,\ast}(J^\infty E; d,\delta,\delta_Q)$; in so doing, the sign ``$\simeq$'' will stand for the equality  modulo $\mathcal{I}\cup d\Lambda$. For example, the Cartan formula (\ref{CMF}) for the Lie derivative along an $\mathcal{F}$-evolutionary vector field $X$ can be written as
$$
L_X\simeq i_X\delta+(-1)^{\tilde{X}}\delta i_X\,.
$$

In the presence of an $m$-dimensional foliation $\mathcal{F}(M)$ it
is quite natural to consider a presymplectic $(2,m)$-form $\omega$
on $J^\infty E$ as  a $\delta$-closed  element of
$\widetilde{{\Lambda}}^{2,m}_\mathcal{F}(J^\infty E)$. Then the
de\-fi\-nitions of the Hamiltonian vector fields and forms should
be modified as follows. An $\mathcal{F}$-evolutionary vector field
$X$ is called $\mathcal{F}$-Hamiltonian if it preserves the
presymplectic form, i.e., $L_X\omega\simeq 0$. An $m$-form $H \in
\widetilde{{\Lambda}}^{0,m}_\mathcal{F}(J^\infty E)$ is called
$\mathcal{F}$-Hamiltonian if there exists an
$\mathcal{F}$-evolutionary vector field $X_H$ such that $\delta H
\simeq  i_{X_H}\omega$. It should be noted  that each Hamiltonian
vector field or $m$-form on $J^\infty E$ is automatically
$\mathcal{F}$-Hamiltonian. Therefore one can regard the space
$\frak{X}_{\omega,\mathcal{F}}(J^\infty E)$ of
$\mathcal{F}$-Hamiltonian vector fields  and the space
${\Lambda}_{\omega, \mathcal{F}}^{0,m} (J^\infty E)$ of
$\mathcal{F}$-Hamiltonian $m$-forms as extensions of the spaces
$\frak{X}_\omega(J^\infty E)$ and
$\widetilde{\Lambda}^{0,m}_\omega(J^\infty E)$, respectively.

Given a presymplectic $(2,m)$-form $\omega$, the space of $\mathcal{F}$-Hamiltonian $m$-forms is equipped with the Lie bracket (\ref{PB}). This bracket is well defined in ${\Lambda}_{\mathcal{F},\omega}^{0,m} (J^\infty E)$ and offers all  the properties of a Lie bracket as one can easily check by repeating the proof of Proposition \ref{2.1} with the sign ``$\simeq$'' meaning now equality modulo $\mathcal{I}\cup d\Lambda$.

\section{BFV from BV}\label{BV-BFV}

In this section, we apply the formalism developed above to establishing a direct correspondence  between the BV formalism of Lagrangian gauge systems and its Hamiltonian counterpart known as the BFV formalism. We start from a very brief account of both the formalisms in a form suitable for our purposes. For a more systematic exposition of the subject we refer the reader to \cite{HT} as well as to the original papers \cite{BV1}-\cite{BFV3}.

\subsection{BV formalism} The starting point of the BV formalism is an infinite-dimensional manifold $\mathcal{M}_0$ of gauge fields that live on an $n$-dimensional space-time $M$. Depending on a particular structure of gauge symmetry the manifold $\mathcal{M}_0$ is extended to an $\mathbb{N}$-graded manifold $\mathcal{M}$ containing $\mathcal{M}_0$ as its body. The new fields of positive $\mathbb{N}$-degree are called the \textit{ghosts} and the $\mathbb{N}$-grading is referred to as the \textit{ghost number}. Let us collectively  denote all the original fields and ghosts by $\Phi^A$ and refer to them as fields. At the next step the space of fields $\mathcal{M}$ is further extended  by introducing the odd cotangent bundle $\Pi T^\ast[-1]\mathcal{M}$. The fiber coordinates, called \textit{antifields}, are denoted by $\Phi_A^\ast$ and assigned with the following ghost numbers and Grassmann parities:
$$
\mathrm{gh} (\Phi^\ast_A)=-\mathrm{gh} (\Phi^A)-1\,,\qquad \epsilon (\Phi^\ast_A)=\epsilon (\Phi^A)+1 \quad (\mbox{mod}\, 2)\,.
$$
Thus, the total space of the odd cotangent bundle $\Pi T^\ast[-1]\mathcal{M}$ becomes a $\mathbb{Z}$-graded supermanifold.  The canonical symplectic structure on $\Pi T^\ast[-1]\mathcal{M}$ is determined  by  the odd $(2,n)$-form
\begin{equation}\label{ops}
\omega= \delta \Phi_A^\ast\wedge \delta\Phi^A\wedge d^nx\,,
\end{equation}
with $d^nx$ being a volume form on $M$. By definition, $\mathrm{gh} (\omega)= - 1$ and $\epsilon (\omega)=1$. The corresponding odd Poisson bracket in the space of functionals of $\Phi$ and $\Phi^\ast$ is given by
\begin{equation}\label{abr}
(A,B)=\int_M \left(\frac{\delta_r A}{\delta \Phi^A}\frac{\delta_l B}{\delta \Phi^\ast_A}-\frac{\delta_r A}{\delta \Phi^\ast_A}\frac{\delta_l B}{\delta \Phi^A}\right)d^nx\,.
\end{equation}
Here the subscripts $l$ and $r$ refer to the standard left and right functional derivatives.
In the physical literature the above bracket  is called usually the \textit{antibracket} or the \textit{BV bracket}.

The central goal of the BV formalism is the construction of a \textit{master action} $S$ on the space of fields and antifields. This is defined as a proper solution to the \textit{classical master equation}
\begin{equation}\label{BV_MEq}
(S,S)=0\,.
\end{equation}
The functional $S$ is required to be of ghost number zero and start with the action $S_0$ of the original fields to which one couples vertices involving antifields. All these vertices can be found systematically from the master equation (\ref{BV_MEq}) by means of the so-called \textit{homological perturbation theory} \cite{HT}. The existence of a proper solution to the classical master equation in the class of local functionals was proved in \cite{H}.

The classical BRST differential on the space of fields and antifields is canonically generated by the master action through the antibracket:
$$
Q=(S\,,\,\cdot\,)\,.
$$
Because of the master equation for $S$ and the Jacobi identity for the antibracket (\ref{abr}),  the operator $Q$ squares to zero in the space of smooth functionals.  The physical quantities are then identified with the cohomology classes of $Q$ in ghost number zero.

\subsection{BFV formalism} The Hamiltonian formulation of the same gauge dynamics implies a prior splitting $M=\mathbb{R}\times N$ of the original space-time into space and time; the factor $N$ can be viewed as the physical space at a given instant of time. The initial values of the original fields are then considered to form an infinite-dimensional manifold $\mathcal{N}_0$. To allow for possible constraints on the initial data of fields the manifold  $\mathcal{N}_0$ is extended to an $\mathbb{N}$-graded supermanifold $\mathcal{N}$ by adding new fields, called ghosts, of positive $\mathbb{N}$-degree. Then the space of original fields and ghosts is doubled by introducing  the cotangent bundle $T^\ast \mathcal{N}$ endowed with the canonical symplectic structure.     If we denote the local coordinates on $\mathcal{N}$ by $\Phi^a$ and the linear coordinates in the cotangent spaces by $\bar{\Phi}_a$, then the canonical symplectic  structure on $T^\ast \mathcal{N}$  is determined by the following $(2,n-1)$-form:
$$
\omega_1=\delta \bar{\Phi}_a\wedge \delta \Phi^a \wedge d^{n-1}x\,.
$$
Here $d^{n-1}x$ stands for a volume form on $N$. By the definition of the cotangent bundle of a graded manifold
$$
{\mathrm{gh}} ( \bar{\Phi}_a) =-\mathrm{gh}({\Phi}^a)\,,\qquad \epsilon(\bar{\Phi}_a) =\epsilon({\Phi}^a)\,,
$$
so that $\omega_1$ is an even $(2,n-1)$-symplectic form of ghost number zero. The corresponding Poisson bracket in the space of functionals of $\Phi^a$ and $\bar{\Phi}_a$ reads
$$
\{A,B\}=\int_N \left(\frac{\delta_r A}{\delta \Phi^a}\frac{\delta_l B}{\delta \bar{\Phi}_a}-(-1)^{{\widetilde{\Phi}}_a}\frac{\delta_r A}{\delta \bar{\Phi}_a}\frac{\delta_l B}{\delta \Phi^a}\right)d^{n-1}x\,.
$$

 The gauge structure of the original dynamics is encoded by the \textit{classical $BRST$ charge} $\Omega$. This is given by an odd functional of ghost number $1$ satisfying the classical master equation
$$
\{\Omega,\Omega\}=0\,.
$$
  Using the method of Ref. \cite{H} one can show that  in any local gauge theory the classical BRST charge can always be constructed as a local functional (see also \cite{KLS2}). The classical BRST differential in the
extended space of fields and momenta is given now by the Hamiltonian action of the BRST charge:
$$
Q=\{\Omega\,,\,\cdot\,\}\,.
$$
It is clear that $Q^2=0$. The group of $Q$-cohomology in ghost number zero is then naturally identified with the space of physical observables.

\subsection{From BV to BFV} It must be clear from the discussion above that any gauge system  in the BFV formalism may be viewed as the descendant of the same system in the BV formalism. More precisely, we can define the even presymplectic structure $\omega_1$ on the phase space of a gauge theory  as the  descendant  of the odd symplectic structure (\ref{ops}):
$$
d\omega_1 =\delta_Q(\delta \Phi^\ast_A\wedge \delta \Phi^A\wedge d^nx)=\delta \left( \delta\Phi^A\wedge \frac{\delta S}{\delta \Phi^A}+\delta\Phi^\ast_A\wedge \frac{\delta S}{\delta \Phi^\ast_A}\right)\,.
$$
Assuming the space-time manifold $M$ to be foliated by the Cauchy hypersurfaces $N\in \mathcal{F}(M)$, we  treat $\omega_1$ as an element of $\widetilde{\Lambda}^{2,n-1}_\mathcal{F}(J^\infty E)$.  The density of the classical BRST charge is then given by the Hamiltonian $J\in \Lambda^{0,n-1}_{\omega_1,\mathcal{F}}(J^\infty E)$  of the classical BRST differential $Q=(S,\,\cdot\,)$ with respect to the descendent presymplectic structure $\omega_1$, that is,
\begin{equation}\label{J}
\delta J\simeq i_Q\omega_1\,.
\end{equation}
 According to Corollary \ref{CL} the form $J$ represents a conserved current, i.e., the BRST current.

Since the canonical symplectic structure (\ref{ops}) on the space of fields and antifields is $\delta$-exact, we can give an equivalent definition for $J$ in terms of the antibracket (\ref{abr}).  For this end, consider the dynamics of fields in a domain $D\subset M$ bounded by  Cauchy hypersurfaces $N_1$ and $N_2$. The fields and antifields are assumed to vanish on space infinity together with their derivatives. By Proposition \ref{p2}, there is an $(n-1)$-form $J$ such that
$$
-\frac12(S,S)=\int_D dJ=\int_{N_2}J-\int_{N_1}J\,.
$$
The classical BRST charge is then given by the space integral\footnote{As in other formulas of this type, it is understood that the local functional $\Omega$ is to be evaluated at a section $\phi$ of $E$ and that the integrand is pulled back to $M$ via $j^\infty \phi$ before being integrated.}
$$
\Omega =\int_N J\,.
$$
It is clear that $\mathrm{gh}(\Omega)=1$. In virtue of Corollary \ref{c1}, the functional $\Omega$ obeys the classical master equation $\{\Omega,\Omega\}=0$ with respect to the even Poisson bracket associated with $\omega_1$.

Let us illustrate the general construction by a particular example of gauge theory.

\subsection{Example: Maxwell's electrodynamics} In the BV formalism, the free electromagnetic field in $4$-dimensional Minkowski space is described by the master action
\begin{equation}\label{MED}
  S=\int L   \,,\qquad L=-\Big(\frac14 F_{\mu\nu}F^{\mu\nu}+C\partial^\mu A_\mu^\ast\Big)d^4x\,.
\end{equation}
Here
$$
F_{\mu\nu}=\partial_\mu A_\nu-\partial_\nu A_\mu
$$
is the strength tensor of the electromagnetic field, $A^\ast_\mu$ is the antifield to the electromagnetic potential $A_\mu$, and $C$ is the ghost field associated with the standard gauge transformation
$
\delta_\varepsilon A_\mu=\partial_\mu \varepsilon
$.

Since the gauge symmetry is abelian, the master action (\ref{MED}) does not involve the ghost antifield $C^\ast$. The odd symplectic structure (\ref{ops}) on the space of fields and antifields assumes the form
$$
\omega=(\delta A^\ast_\mu\wedge \delta A^\mu+\delta C^\ast\wedge \delta C)\wedge d^4 x\,,\qquad d^4x=dx^0\wedge dx^1\wedge dx^2\wedge dx^3\,,
$$
and the action of the classical BRST differential is given by
\begin{equation}\label{brst-t}
\delta_Q A_\mu =\partial_\mu C\,,\qquad \delta_Q A^\ast_\mu =\partial^\nu F_{\nu\mu}\,,\qquad \delta_Q C=0\,, \qquad \delta_Q C^\ast=\partial^\mu A^\ast_\mu\,.
\end{equation}
The variation of the Lagrangian  density reads
$$
\delta L=(\partial^\mu F_{\mu\nu}\delta A^\nu +  \partial^\mu A^\ast_\mu\delta C +\partial^\mu C\delta A^\ast_\mu  -\partial^\mu \theta_\mu)\wedge d^4x\,,\qquad \theta_\mu =
F_{\mu\nu}\delta A^\nu +C \delta A_\mu^\ast\,.
$$
One can easily check that $i_Q \omega \simeq \delta L$.  By Proposition \ref{p5} the form
$$
\theta_1=-\theta_\mu\wedge d^3x^\mu\,,\qquad d^3x^\mu=\eta^{\mu\nu}i_{\frac{\partial}{\partial x^\nu}}d^4x\,,$$
defines the potential for the descendent presymplectic form
\begin{equation}\label{WC}
\omega_1=\delta\theta_1= -(\delta F_{\nu\mu}\wedge \delta A^\mu+\delta C\wedge \delta A^\ast_\nu)\wedge d^3x^\nu\,.
\end{equation}
(Of course, one  could arrive at the same expression by considering the BRST variation $\delta_Q\omega = d\omega_1$ of the original symplectic structure.) Except for the ghost term the covariant presymplectic structure  (\ref{WC}) for the electromagnetic field  was first introduced in \cite{CW}.

Applying the BRST differential to the form $\omega_1$ yields one more descendent presymplectic form
$$
\omega_2= \delta C\wedge \delta F_{\mu\nu}\wedge d^2x^{\mu\nu}\,, \qquad d^2x^{\mu\nu}=\eta^{\mu\alpha}i_{\frac{\partial}{\partial x^\alpha}}d^3x^\nu\,.
$$
This last form, being ``absolutely'' invariant under the BRST transformations (\ref{brst-t}), leaves no further descendants.

The $3$-form of the conserved BRST current $J$ associated to the BRST symmetry transformations (\ref{brst-t}) is determined by Eq. (\ref{J}). We find
$$
J=J_\nu d^3x^\nu\simeq-C\partial^\mu F_{\mu\nu} d^3x^\nu\,.
$$

In order to obtain the BRST charge and the presymplectic form on the phase space of the theory we need to fix a causal structure on $\mathbb{R}^{3,1}$.  Identifying the coordinate $x^0$ with a global time in the Minkowski space, we set $\eta=dx^0$. The leaves $N$ of the corresponding foliation $\mathcal{F}(\mathbb{R}^{3,1})$ are given by the space-like hyperplanes $x^0=c$.  Then the descendent presymplectic structure is represented by the $(2,3)$-form
\begin{equation}\label{br-em}
\omega_1=-(\delta F_{0 i}\wedge \delta A^i +\delta A_0^\ast\wedge \delta C)\wedge d^3x\,,
\end{equation}
$$
d^3x=dx^1\wedge dx^2\wedge dx^3\,,\qquad i=1,2,3\,.
$$
and the density of the classical BRST charge is represented by the $3$-form
$$
J_0 d^3x\simeq -C
\partial^i F_{i0}d^3x\,.
$$
We see that the zero component of the antifield $A^\ast_\mu$ plays the role of the ghost momentum $\bar{\mathcal{P}}$ canonically conjugate to $C$ and
 the role of the canonical momentum to the $3$-vector $A^i$ is played by the $3$-vector of electric field $E_i=F_{i0}$. The on-shell conservation of the BRST charge $\Omega=\int_N J_0d^3x$ expresses nothing but the Gauss law $\partial^i E_i=0$. One can also see that  $\Omega$ satisfies the classical master equation $\{\Omega,\Omega\}=0$ with respect to the canonical Poisson bracket on the phase space of fields $(A^i, E_i; C, \bar{\mathcal{P}})$.
Furthermore, any $3$-form
$$
f(\partial_I A^i, \partial_I E_i, \partial_IC, \partial_I \bar{\mathcal{P}})d^3x\,,\qquad I=i_1i_2\cdots i_s\,,
$$
appears to be $\mathcal{F}$-Hamiltonian with respect to (\ref{br-em}), so that the Lie algebra of $\mathcal{F}$-Hamiltonian forms is rich enough.  In particular, it includes the energy density of the electromagnetic field
\begin{equation}\label{H}
H=\frac12\Big( E_iE^i+\frac12F_{ij}F^{ij}\Big)d^3x\,.
\end{equation}
It is clear that $\{H,\Omega\}=0$. The last equation implies two things: (i) the physical energy is BRST invariant  and (ii) the BRST charge is invariant with respect to the time evolution generated by the physical Hamiltonian (\ref{H}).

\section{Weak Poisson bracket from the Lagrange structure}

We start with a brief review of the BRST theory of non-variational gauge systems as it was first formulated in \cite{LS0} and \cite{KazLS}.

\subsection{Non-Lagrangian gauge systems and $S_\infty$-algebras} The geometrical arena for the BRST formulation of not necessarily Lagrangian gauge dynamics is provided by the cotangent bundle $T^\ast \mathcal{M}$ of the space of fields $\mathcal{M}$. As before, we consider $\mathcal{M}$ to be an infinite-dimensional $\mathbb{N}$-graded manifold parameterized locally  by some set of field $\Phi^A$ over an $n$-dimensional space-time manifold $M$.  The fields with nonzero $\mathbb{N}$-degree (= ghost number) are called collectively the ghosts, while the original  fields are characterized by ghost number zero. The linear coordinates in the fibres of  $T^\ast \mathcal{M}$ are denote by $\bar{\Phi}_A$; in physical terms, they have the meaning of \textit{sources} to the fields $\Phi^A$. According to the definition of the cotangent bundle of  a graded manifold
$$
\mathrm{gh}(\bar{\Phi}_A)=-\mathrm{gh}({\Phi}^A)\,,\qquad \epsilon(\bar{\Phi}_A)=\epsilon({\Phi}^A)\,.
$$
In such a way the total space of the cotangent bundle $T^\ast \mathcal{M}$ becomes a $\mathbb{Z}$-graded manifold.
The cotangent bundle carries the canonical symplectic structure defined by the $(2,n)$-form
\begin{equation}\label{w-s}
\omega =\delta \bar{\Phi}_A\wedge \delta\Phi^A \wedge d^nx\,,
\end{equation}
with $d^nx$ being a volume form on $M$. Contrary to the BV formalism this symplectic structure is even and has ghost number zero. The corresponding Poisson bracket on the phase space of fields and sources is given by
$$
\{A,B\}=\int _M \left(\frac{\delta_r A}{\delta \Phi^A}\frac{\delta_l B}{\delta \bar{\Phi}_A}-(-1)^{{\widetilde{\Phi}}_A}\frac{\delta_r A}{\delta \bar{\Phi}_A}\frac{\delta_l B}{\delta \Phi^A}\right)d^{n}x\,,
$$
where $A$ and $B$ are functionals of $\Phi$'s and $\bar\Phi$'s.  It is clear that for local functionals $A=\int_M a$ and $B=\int_M b$ this Poisson bracket is given by the integral over $M$ of the Poisson bracket $\{a, b\}$ of two Hamiltonian $n$-forms $a,b\in {\Lambda}_{\omega}^{0,n}(J^\infty E) $.

 Besides the Grassmann parity and the ghost number the manifold $T^\ast \mathcal{M}$ is endowed with one more ${\mathbb{N}}$-grading called the \textit{momentum degree}. This is introduced  by prescribing the following degrees to the fields and sources:
$$
\mathrm{Deg} (\Phi^A)=0\,,\qquad \mathrm{Deg} (\bar{\Phi}_A)=1\,.
$$
In the context of local field theory this grading can also be conveniently described by means of the \textit{Euler vector field}
$$
E_m=\sum_{|I|=0}^\infty \bar{\Phi}_{AI}\frac{\partial}{\partial \bar \Phi_{AI}}\,,\qquad E_m\in \mathfrak{X}_{ev}(J^\infty E)\,.
$$
A form $H\in \Lambda^{p,q}(J^\infty E)$ is said to be homogeneous of momentum degree $k$ iff
$$
L_{E_m} H = k H \,,\qquad k =\mathrm{Deg}\, H  \,.
$$
In particular, the symplectic form (\ref{w-s}) is homogeneous of momentum degree $1$. The momentum degree of a homogeneous vector field $X\in \mathfrak{X}(J^\infty E)$ is defined in similar manner:
$$
L_{E_m}X=[E_m,X]=kX\,,\qquad k=\mathrm{Deg}\,X\,.
$$

A (non-)Lagrangian gauge system is completely specified by a \textit{total BRST charge}.  This is given by a local functional $\Omega$ of fields and sources satisfying the classical master  equation
\begin{equation}\label{meq-s}
\{\Omega, \Omega\}=0
\end{equation}
and the grading conditions
$$
\epsilon(\Omega)=1\,,\qquad \mathrm{gh} (\Omega)=1\,,\qquad \mathrm{Deg} (\Omega) >0\,.
$$
The last inequality implies the following expansion for $\Omega$ according to the momentum degree:
\begin{equation}\label{o-exp}
\Omega=\sum_{k=1}^\infty \Omega_k\,,\qquad \mathrm{Deg }(\Omega_k)=k\,.
\end{equation}
In other words, $\Omega|_\mathcal{M}=0$. In terms of this expansion  the classical master equation (\ref{meq-s}) is equivalent to the infinite sequence of relations
$$
\{\Omega_1,\Omega_1\}=0\,,\qquad \{\Omega_1,\Omega_2\}=0\,,\qquad \{\Omega_2,\Omega_2\}=2\{\Omega_1,\Omega_3\}\,,\qquad \ldots
$$

As is seen the leading term of the expansion (\ref{o-exp}), called the \textit{classical BRST charge}, is Poisson-nilpotent by itself. It carries all information about the original classical system. To be more specific, let us write a few first terms of $\Omega_1$ by making use of DeWitt's condensed index notation\footnote{According to this notation, the discrete indices labeling the components of fields include also space-time coordinates and summation by a pair of repeated indices implies integration over the space-time.}:
\begin{equation}\label{o1}
\Omega_1=T_a(\phi)\bar{\eta}^a +c^\alpha R_\alpha^i(\phi)\bar{\phi}_i+\eta_a L^a_A(\phi)\bar{\xi}^A+\cdots\,.
\end{equation}
Here $T_a(\phi)=0$ are the equations of motion for the original fields $\phi^i$; $\eta$, $\xi$, ... are the ghost fields, and the variables with bar denote the corresponding sources. The $R$'s are determined by the generators of the gauge transformations $\delta_{\varepsilon}\phi^i=R^i_\alpha\varepsilon^\alpha$, while $L$'s generate  the gauge identities $L^a_AT_a(\phi)=0$ for the field equations. Notice that beyond the setting of Lagrangian field theory the generators $R$'s and $L$'s are completely unrelated to each other.

The second term in (\ref{o-exp}) incorporates the so-called \textit{Lagrange structure} \cite{KazLS}.  Its expansion in terms of ghosts and sources starts as
$$
\Omega_2=\bar{\eta}^aV_a^i(\phi)\bar{\phi}_i+\eta_aC^a_{bc}(\phi)\bar{\eta}^b\bar{\eta}^c+\cdots\,.
$$
The structure coefficients $V^i_a$ entering the leading term of $\Omega_2$ are called the \textit{Lagrange anchor}. The choice of a Lagrange anchor determines the way in which the classical gauge system (\ref{o1}) is supposed to be quantized. Different Lagrange anchors lead generally to inequivalent quantum theories with the  same classical limit.

In order to further elucidate the physical meaning of expansion (\ref{o-exp}) and its relation to the standard BV formalism, it is convenient to introduce the commutative algebra $\mathcal{A}$ of smooth functionals of momentum degree zero. In \cite{KazLS}, it was observed that each total BRST charge (\ref{o-exp}) gives $\mathcal{A}$ the structure of a flat $S_\infty$-algebra \cite{Vor2}. The corresponding structure maps $S_n: \mathcal{A}^{\otimes n}\rightarrow \mathcal{A}$ are defined through the derived bracket construction
\begin{equation}\label{hs}
S_n: a_1\otimes a_2\otimes \cdots \otimes a_n \quad \mapsto\quad (a_1,a_2,\ldots,a_n)=\{\cdots \{\Omega_n, a_1\},a_2\},\cdots,a_n\}\in \mathcal{A}\,.
\end{equation}
It follows from the definition that
\begin{itemize}
  \item[(a)] the multi-brackets are odd and symmetric,
$$
(a_1,\ldots, a_k,a_{k+1},\ldots, a_n)=(-1)^{\tilde{a}_k \tilde{a}_{k+1}}(a_1,\ldots,a_{k+1},a_k,\ldots, a_n)\,,$$
  \item[(b)] the map
$$
a\quad \mapsto\quad (a_1,\ldots,a_{n-1},a)
$$
is a derivation of $\mathcal{A}$ of parity $1+\sum_{k=1}^{n-1}\tilde{a}_k$  $(\mathrm{mod}\; 2)$ ,
  \item[(c)] for all $m\geq 2$
  $$
\sum_{k+l=m} \sum_{(k,l)-shuffles} (-1)^\epsilon ((a_{\sigma(1)},\ldots,a_{\sigma(k)}),a_{\sigma(k+1)},\ldots, a_{\sigma(k+l)})=0\,,
$$
where $(-1)^\epsilon$ is a natural sign prescribed by the sign rule for permutations of homogeneous elements $a_1,\ldots,a_n\in \mathcal{A}$, and the $(k,l)$-shuffle is a permutation of indices $1,2,\ldots,k+l$ satisfying $\sigma(1)<\cdots<\sigma(k)$ and $\sigma(k+1)<\cdots,\sigma(k+l)$.
\end{itemize}

One can check \cite{Vor2}, that the set of the generalized Jacobi identities of item (c) is just another form of the component master equations (\ref{o-exp}). For $m=1,2,3$ these identities take the form
\begin{equation}\label{hom}
\begin{array}{c}
Q^2 a = 0\,,\\[3mm]
Q(a, b) + (Qa, b) + (-1)^{\tilde{a}\tilde{b}}
(Qb, a) = 0\,,\\[3mm]
((a, b), c) + (-1)^{\tilde{b}\tilde{c}}˜
((a, c), b) + (-1)^{\tilde{a}˜(˜\tilde{b}+˜\tilde{c})}
((b, c), a)
\\[3mm]
+Q(a, b, c) + (Qa, b, c) + (-1)^{\tilde{a}˜\tilde{b}}(Qb, a, c) + (-1)
^{(\tilde{a}+˜
\tilde{b})˜\tilde{c}}
(Qc, a, b) = 0\,,
\end{array}
\end{equation}
 where we have denoted the unary bracket by $Q=(\,\cdot\,)$. The operator $Q: \mathcal{A}\rightarrow \mathcal{A}$  squares to zero and has the property of odd derivation w.r.t. the binary bracket as well as the commutative multiplication  (item (b) above). The operator $Q$, being given by the Hamiltonian action of the classical BRST charge $\Omega_1$,  is naturally identified with the classical BRST differential of the gauge system.

The binary bracket is defined by the Lagrange structure $\Omega_2$. According to the last equality in (\ref{hom}) the bracket obeys the Jacobi identity up to the homotopy correction determined by the trinary bracket. The Jacobi identities with $m>3$ impose extra relations on this homotopy and all the higher homotopies. Taken together Eqs. (\ref{hom}) mean that the binary bracket induces the usual odd Poisson bracket in the cohomology of the differential $Q$.

In a particular case, where the expansion of the total BRST charge (\ref{o-exp})  stops at the second
term,  $\Omega = \Omega_1 + \Omega_2$, the binary  bracket
$$
(a,b)=(-1)^{\tilde{a}}\{\{\Omega_2,a\},b\}
$$
enjoys all the properties of the ordinary BV bracket (\ref{abr}), including the Jacobi identity. If we further assume  the binary bracket to be non-degenerate, then the classical BRST differential associated with $\Omega_1$ is necessarily given by the Hamiltonian vector field $Q = (S, \,\cdot \, )$, with $S$ being the master action. The classical master equation $(S, S) = 0$ follows  from the identity $Q^2=0$. This explains how the standard BV formalism for Lagrangian gauge systems fits in this more general quantization
approach.

\subsection{Weakly Hamiltonian systems and $P_\infty$-algebras} The phase-space counterpart of the Lagrange structure above is known as the \textit{weak Poisson structure} \cite{LS0}. This is defined on the odd cotangent bundle of an  $\mathbb{N}$-graded manifold $\mathcal{N}$. The latter represents an extension by ghost variables of the space of initial data to the field equations.
Locally, the manifold $\mathcal{N}$ is parameterized by the set of fields $\Phi^a$ leaving on an $(n-1)$-dimensional physical space $N$. The canonical sympletic structure on $\Pi T^\ast[1] \mathcal{N}$ is determined by  the following $(2,n-1)$-form:
\begin{equation}\label{p-sym}
  \omega_1=\delta {\Phi}^\ast_a\wedge \delta\Phi^a \wedge d^{n-1}x\,.
\end{equation}
Here $\Phi^\ast_a$ are linear coordinates in the fibers of $\Pi T^\ast[1]{\mathcal{N}}$.
Following the physical tradition, we refer to $\Phi^\ast_a$ as antifields. By the definition of a shifted, odd cotangent bundle
$$
  \mathrm{gh}(\Phi^\ast_a)=-\mathrm{gh}(\Phi^a)+1\,,\qquad \epsilon(\Phi^\ast_a) =\epsilon(\Phi^a)+1 \quad (\mbox{mod}\, 2)\,.
$$
The corresponding odd Poisson bracket in the space of fields and antifields reads
$$
(A,B)=\int_M \left(\frac{\delta_r A}{\delta \Phi^a}\frac{\delta_l B}{\delta \Phi^\ast_a}-\frac{\delta_r A}{\delta \Phi^\ast_a}\frac{\delta_l B}{\delta \Phi^a}\right)d^{n-1}x\,.
$$

The manifold $\Pi T^\ast[1] \mathcal{N}$ can be endowed with a natural $\mathbb{N}$-grading by prescribing the following degrees to the local coordinates:
$$
\mathrm{Deg} (\Phi^a)=0\,,\qquad  \mathrm{Deg}(\Phi^\ast_a)=1\,.
$$
This additional grading is called the \textit{polyvector degree}. Again, in the local setting, one can define this grading  through the Euler vector field on the jet space
$$
E_p=\sum_{|I|=0}^\infty \Phi_{aI}^\ast\frac{\partial }{\partial \Phi^\ast_{aI}}\,, \qquad E_p\in \mathfrak{X}_{ev}(J^\infty E)\,,
$$
so that the homogeneous forms or vector fields on $J^\infty E$ belong to a definite eigenvalue of the operator  $L_{E_p}$.  In particular, $L_{E_p}\omega_1=\omega_1$, which means that the odd symplectic form (\ref{p-sym}) is homogeneous of polyvector degree $1$.

A weak Poisson structure on $\mathcal{N}$ is determined by a local functional $S$ satisfying the master   equation
$$
(S,S)=0
$$
and the gradding conditions
$$
\epsilon(S)=0\,,\qquad \mathrm{gh}(S)=2\,,\qquad \mathrm{Deg}(S)>0\,.
$$
Notice that contrary to the BV formalism the functional $S$, being even, has ghost number $2$. The physical meaning of the functional $S$ becomes clear if one expands it in powers of ghosts and antifields. By making use of DeWitt's condensed notation we can write\
\begin{equation}\label{s-exp}
S=\sum_{k=1}^\infty S_k\,,\qquad \mathrm{Deg} \, S_k=k\,,
\end{equation}
where
$$
S_1=T_a(\phi)\eta^{\ast a}+ c^\alpha R_\alpha^i(\phi)\phi_i^\ast+\cdots\,, \qquad S_2=P^{ij}(\phi)\phi^\ast_i\phi^\ast_j+\cdots\,.
$$
The structure coefficients of this expansion have the following interpretation. The equations $T_a(\phi)=0$ are the Hamiltonian constraints on the original phase-space variables $\phi^i$. The $R$'s are given by the generators of the gauge symmetry transformations $\delta \phi^i=R^i_\alpha \varepsilon^\alpha$. The leading term of $S_2$ defines a weak Poisson structure on the phase space of fields $\phi^i$ determined by the bivector $P=P^{ij}(\phi)\partial_i\wedge \partial_j$.

Expansion (\ref{s-exp}) for the generating functional  $S$ admits also  a straightforward interpretation in terms of $P_\infty$-algebras \cite{Vor2}: we let $\mathcal{A}$ denote the commutative algebra of functionals of polyvector degree $0$ and define the $n$-th structure map $P_n: \mathcal{A}^{\otimes n}\rightarrow \mathcal{A}$ through the derived bracket
\begin{equation}\label{hp}
P_n: a_1\otimes a_2\otimes \cdots \otimes a_n \quad \mapsto\quad \{a_1,a_2,\ldots,a_n\}=(\cdots (S_n, a_1),a_2),\cdots,a_n)\in \mathcal{A}\,.
\end{equation}
It follows from the definition that
\begin{itemize}
  \item[(a)] the multi-brackets are even and satisfy the symmetry property
$$
\{a_1,\ldots, a_k,a_{k+1},\ldots, a_n\}=(-1)^{(\tilde{a}_k+1)( \tilde{a}_{k+1}+1)}\{a_1,\ldots,a_{k+1},a_k,\ldots, a_n\}\,,$$
  \item[(b)] the map
$$
a\quad \mapsto\quad \{a_1,\ldots,a_{n-1},a\}
$$
is a derivation of $\mathcal{A}$ of parity $\sum_{k=1}^{n-1}\tilde{a}_k$  $(\mathrm{mod}\; 2)$ ,
  \item[(c)] for all $m\geq 2$
  $$
\sum_{k+l=m} \sum_{(k,l)-shuffles} (-1)^\epsilon \{\{a_{\sigma(1)},\ldots,a_{\sigma(k)}\},a_{\sigma(k+1)},\ldots, a_{\sigma(k+l)}\}=0\,,
$$
where $(-1)^\epsilon$ is some natural  sign factor, see \cite{Vor2}, \cite{B}.
\end{itemize}

As with $S_\infty$-algebras, the unary bracket defines the classical BRST differential $Q$. The binary bracket is differentiated by $Q$ and satisfies the Jacobi identity up to homotopy corrections controlled by the trinary bracket and $Q$. Moreover, the binary bracket descends  to the cohomology of $Q$ inducing a true Poisson bracket in the space of physical observables.

It should be noted that contrary to the classical BRST charge (\ref{o-exp}), the functional $S$ is only a part of data needed to formulate the gauge dynamics. The other part is given by the local functional $\Gamma$ obeying the following conditions \cite{LS0}:
$$
(S,\Gamma)=0\,,\qquad \epsilon(\Gamma)=1\,,\qquad \mathrm{gh}(\Gamma)=1\,,\qquad \mathrm{Deg}\,(\Gamma)>0\,.
$$
It is the functional $\Gamma$ that generates the time evolution of the system through the odd Poisson bracket:
$$
\dot \Phi^a= (\Gamma, \Phi^a)\,,\qquad \dot \Phi_a^\ast=(\Gamma, \Phi_a^\ast)\,.
$$
Taken together the functionals $S$ and $\Gamma$ define a \textit{weak Hamiltonian structure}, which can be
regarded as a strong homotopy generalization of the conventional BFV formalism.  The latter corresponds to a special case where $S=S_1+S_2$ and the even Poisson bracket associated to $S_2$ is non-degenerate.

\subsection{$P_\infty$ from $S_\infty$} In the above discussion of $S_\infty$- and $P_\infty$-algebras we assumed $\mathcal{A}$ to be the commutative algebra of smooth functionals on $\mathcal{M}$ and $\mathcal{N}$, respectively.
As a linear  space  the algebra $\mathcal{A}$ contains the subspace  of local functionals $\mathcal{A}_{loc}$. The latter space represents the main interest for the local field theory. Although $\mathcal{A}_{loc}$ is not a commutative subalgebra, it is still closed  with respect to the multi-brackets (\ref{hs}) and (\ref{hp}). In mathematics, a linear space endowed with a sequence of multi-brackets satisfying the generalized Jacobi identities is known as an $L_\infty$-algebra \cite{LSt}. Notice that the $L_\infty$-algebras underlying the Lagrange and the weak Poisson structures differ by the parity and symmetry properties of the corresponding multi-brackets. This difference, however, is not fundamental as one can switch between the two definitions by applying the parity reversion functor to $\mathcal{A}$, see \cite{Vor2}, \cite{B}. In what follows we consider only local functionals and their $L_\infty$-algebras.

Suppose we are given a total BRST charge
$$\Omega=\int_M O\,,\qquad O\in {\Lambda}^{0,n}(J^\infty E)\,,$$
associated with some (non-)variational gauge system, and let $\mathcal{F}(M)$  be an $(n-1)$-dimensional foliation of $M$ by the Cauchy hypersurfaces. The master equation
(\ref{meq-s}) for the BRST charge $\Omega$ is equivalent to
$$
-\frac12\{O,O\}=d\Sigma
$$
for some form $\Sigma\in {\Lambda}{}^{0,n-1}(J^\infty E)$ of ghost number $2$; here the braces denote the Poisson bracket in $\Lambda^{0,n}_{\omega} (J^\infty E)$ coming from the canonical symplectic form (\ref{w-s}).

Let $\mathcal{Q}$ denote the Hamiltonian vector field corresponding to  $O$. We have
$$
-i_{\mathcal{Q}^2}\omega\simeq\frac12 \delta \{O,O\}\simeq 0\,.
$$
Since $i_{\mathcal{Q}^2}\omega$ is a source $1$-form, the last relation implies the ``strong'' equality  $i_{\mathcal{Q}^2}\omega=0$ (Proposition \ref{SF}) and from the non-degeneracy of $\omega$ it then follows that $\mathcal{Q}^2=0$.  Thus, $\mathcal{Q}$ is a homological vector field and we can consider the descendent gauge system $(\mathcal{Q}, \omega_1)$.
 By Proposition \ref{p5}, the presymplectic potential for the descendent presymplectic form $\omega_1=\delta\theta_1$ is determined by the variation of the total BRST charge, namely,
$$
\delta O =\delta\Phi^A\wedge \frac{\delta O}{\delta \Phi^A} +\delta\bar{\Phi}_A \wedge\frac{\delta O}{\delta \bar{\Phi}_A}- d\theta_1\,.
$$
According to Proposition (\ref{p2}), the homological vector field $\mathcal{Q}$ is Hamiltonian with respect to
$\omega_1$ with the Hamiltonian given by the $(n-1)$-form $\Sigma$. Denoting by $(\,\cdot\,,\,\cdot\,)$ the Poisson bracket in the space of $\mathcal{F}$-Hamiltonian $(n-1)$-forms, we can write
$$
(\Sigma,\Sigma)\simeq 0\,,
$$
see Corollary \ref{c1}.
The preceding discussion of the relationship between the BV and BFV formalisms makes it quite  reasonable to identify the integral
$$
S=\int_C \Sigma\,,\qquad C\in \mathcal{F}(M) \,,
$$
with the generating functional of a weak Poisson structure; in so doing, we need to specify the polyvector degree of $\mathcal{F}$-Hamiltonian forms. A momentary reflection shows that the naive identification of the polyvector degree with the momentum degree does not work in general. The reason is obvious: the presymplectic form $\omega_1$ is not homogeneous with respect to the momentum degree, rather it is given by the sum
\begin{equation}\label{om1}
\omega_1=\sum_{k=1}^\infty\omega_1^{_{(k)}}\,,\qquad \mathrm{Deg}(\omega_1^{_{(k)}})=k\,,
\end{equation}
where $\omega_1^{_{(1)}}$ is determined by the classical BRST charge $\Omega_1$, $\omega_1^{_{(2)}}$ comes from the Lagrange structure $\Omega_2$, and so on. As a result, the $\mathcal{F}$-Hamiltonian forms of momentum degree zero do not form a commutative algebra with respect to the odd Poisson bracket $(\,\cdot\,,\,\cdot\,)$ whenever $\omega_1^{(k)}\neq 0$ for $k>1$. This makes impossible the derived bracket construction and the interpretation of the homogeneous components of $\Sigma$ as multi-brackets on the space of $\mathcal{F}$-Hamiltonian forms of momentum degree zero.

In order to overcome this difficulty and equip the space of fields and sources with an appropriate polyvector degree we will assume that the presymplectic form (\ref{om1}) is homotopically  equivalent to its leading term $\omega_1^{_{(1)}}$. More precisely, there must exist a formal diffeomorphism of the jet space $h=e^X: J^\infty E\rightarrow J^\infty E$ generated by a vertical vector field $X$ with $\mathrm{Deg}\, X>0$ such that
$$
h^\ast(\omega_1^{_{(1)}})\simeq \omega_1\,.
$$
Then we define the Euler field counting the polyvector degree as
\begin{equation}\label{pdeg}
E_p= h^\ast (E_m)=e^{X} E_m e^{-X}\,,
\end{equation}
where
$$
E_m=\sum_{|I|=0}^\infty \bar\Phi_{AI}\frac{\partial}{\partial \Phi_{AI}}
$$
is the evolutionary vector field counting the momentum degree.
By definition we have
$$
L_{E_p}\omega_1=\omega_1\,.
$$
 This means that  $\omega_1$ is an odd presymplectic form of polyvector degree 1. Then the polyvector degree of the bracket  $(\,\cdot\,,\,\cdot\,)$, being opposite to that of $\omega_1$,   is equal to $-1$.  It remains to note that the $\mathcal{F}$-Hamiltonian $(n-1)$-forms of polyvector degree zero form a commutative algebra with respect to the odd bracket above. For if the bracket  $(a,b)$ of two such forms were  nonzero, it would be of polyvector degree $-1$.

\begin{rem}
No example of gauge theory where the derived presymplectic form $\omega_1$ would not be homotopical to its homogeneous part $\omega^{_{(1)}}_1$ has come to our notice. It seems that the homotopical equivalence $\omega_1 \sim\omega^{_{(1)}}_1$ is a general property ensured by the properness of the classical BRST differential, although we  have not a complete proof of this fact at the moment.

 Given a presymplectic form of polyvector degree $1$, we can equip the space of $\mathcal{F}$-Hamiltonian forms of polyvector degree zero with the structure of a $P_\infty$-algebra. The corresponding multibrackets are given by
\begin{equation}\label{binbr}
\{a_1,a_2,\ldots,a_n\}=(\cdots(\Sigma_n,a_1), a_2,\cdots, a_n)\,,\qquad n=1,2, \ldots\,,
\end{equation}
where
$$
\Sigma =\sum_{n=1}^\infty \Sigma_n\,,\qquad L_{E_p}\Sigma_n=n\Sigma_n\,,\qquad L_{E_p}a_i=0\,.
$$
As usual the unary bracket defines the action of the classical BRST differential $Q$.  It follows from the definition of polyvector degree (\ref{pdeg}) that
$$
\{a\}=(a, \Sigma_1)=L_Qa\,,\qquad (b)=\{b, O_1\}=L_Q b
$$
for all $a\in  {\Lambda}^{0,n-1}_{\omega_1, \mathcal{F}}(J^\infty E)$ and $b\in {\Lambda}^{0,n}_{\omega}(J^\infty E)$. In either formalism,
the classical BRST differential is given by one and the same homological vector field $Q$ on $J^\infty E$. Then the binary bracket in (\ref{binbr}),
$$
\{a,b\}=(-1)^{\tilde{a}}((\Sigma_2,a),b)\,,
$$
is differentiated by $Q$ and induces a Lie bracket in the cohomology of $Q$. In such a way the space of physical observables -- BRST invariant, local functionals in ghost number zero -- gets  the structure of a Lie algebra.

The results of this subsection can be summarized as follows. Starting from the total BRST charge associated to a (non-)variational gauge system endowed with a Lagrange structure, we were able to define the Lie bracket in the space of physical observables  under the assumption that the descendent presymplectic structure is homotopically equivalent to its momentum-degree-one  part.
The last property is not too restrictive and holds for all known examples of Lagrange structures. This Lie bracket can then be extended, in natural way, to the Poisson bracket in the algebra of non-local functionals that are BRST invariant and sufficiently smooth.

\end{rem}

\subsection{Example: chiral bosons in two dimensions}
Let $\mathbb{R}^{1,1}$  be a two-dimensional Min\-kow\-ski space with the light-cone coordinates $x^{\pm}=\tau\pm \sigma$ and let $\Lambda(\mathbb{R}^{1,1})\otimes \mathcal{G}$ denote the space of differential forms on $\mathbb{R}^{1,1}$ with values in a semisimple Lie algebra $\mathcal{G}$. In the light-cone frame  the space of $1$-forms splits into the direct sum $\Lambda_+^1(\mathbb{R}^{1,1})\oplus \Lambda_-^1(\mathbb{R}^{1,1})$ of two subspaces spanned, respectively, by the self-dual and anti-self-dual forms.  By the triangle brackets $\langle\,\cdot\,, \,\cdot\,\rangle$ we will denote the invariant bilinear form on $\mathcal{G}$. This bilinear form as well as the commutator $[\,\cdot\,\,,\cdot\,]$ in $\mathcal{G}$ are naturally extended  to the tensor product  $\Lambda(\mathbb{R}^{1,1})\otimes \mathcal{G}$.

Now consider the relativistic-invariant field equations
\begin{equation}\label{df}
d\phi=0\,,
\end{equation}
where $\phi\in \Lambda^1_+(\mathbb{R}^{1,1})\otimes \mathcal{G}$. If $\{t_i\}$ is a basis in $\mathcal{G}$, then $\phi=t_i\phi^i_+dx^+$ and the equations of motion take the form
$
\partial_-\phi_+^i=0
$.
These equations are known to be non-Lagrangian but admit a one-parameter family of nontrivial Lagrange structures \cite{KLS1}, \cite{Sh}. The corresponding BRST charge reads
$$\Omega=\int_{\mathbb{R}^{1,1}} O\,,\qquad O=O_1+O_2\,,
$$
\begin{equation}\label{OOO}
O_1=\langle\bar\eta, d\phi\rangle\,,\qquad O_2=\langle\bar\eta,[\phi,\bar\phi]+kd\bar\phi\rangle+\frac12\langle\eta,[\bar\eta,\bar\eta]\rangle\,,\qquad k\in \mathbb{R}\,.
\end{equation}
Here
$$
\begin{array}{ccc}
\bar \phi=\bar\phi_-dx^{-}\in \Lambda^1_-(\mathbb{R}^{1,1})\otimes \mathcal{G}\,,&\quad \bar\eta\in \Lambda^0(\mathbb{R}^{1,1})\otimes \mathcal{G}\,,&\quad \eta=\eta_{+-}dx^+\wedge dx^-\in \Lambda^2(\mathbb{R}^{1,1})\otimes \mathcal{G}\,,\\[2mm]
\mathrm{gh}(\bar\phi)=0\,,&\quad \mathrm{gh}(\bar\eta)=1\,,&\quad \mathrm{gh}(\eta)=-1\,.
\end{array}
$$
The form $O_1$ defines the classical BRST charge. In the absence of gauge symmetries and identities, it is constructed in terms of the equations of motion (\ref{df}) alone. The Lagrange structure defines and is defined by the form $O_2$.

The canonical symplectic structure in the space of fields and sources is determined by the $(2,2)$-form
$$
\omega=\langle\delta \bar\phi,\delta\phi\rangle+\langle\delta\bar\eta,\delta\eta\rangle\,.
$$
Evaluating the master equation for the total BRST charge, we find
$$
\frac12\{O,O\}=d\Sigma \,,\qquad \Sigma=\Sigma_2+\Sigma_3\,,
$$
$$
\Sigma_2=\langle \phi, [\bar\eta,\bar\eta]\rangle+k\langle \bar\eta, d\bar\eta\rangle\,,\qquad \Sigma_3=k\langle \bar\phi,[\bar\eta,\bar\eta]\rangle\,.
$$

By virtue of Proposition \ref{p5}, the presymplectic $(2,1)$-form $\omega_1$ of the descendent gauge system is determined by the variation of the total BRST charge, namely,
$$
\delta O = \delta \Phi^A\wedge \frac{\delta O}{\delta \Phi^A} + \delta\bar{\Phi}_A\wedge \frac{\delta O}{\delta \bar{\Phi}_A}- d\theta_1\,, \qquad \Phi^A=(\phi^i_+, \eta^i)\,,\quad \bar{\Phi}_A=(\bar\phi^i_-, \bar\eta^i_{+-})\,,
$$
$$
\omega_1=\delta \theta_1 = \omega_1^{(1)}+\omega_1^{(2)}=\langle\delta \bar\eta,\delta\phi\rangle+k\langle\delta \bar\eta,\delta\bar\phi\rangle\,.
$$
As is seen the descendent presymplectic form $\omega_1$ is inhomogeneous with respect to the momentum degree.

Now to define a (weak) Poisson bracket in the phase space of fields $\phi$ we have to fix a causal structure on $\mathbb{R}^{1,1}$. This is given by the one-dimensional Cauchy foliation $\mathcal{F}(\mathbb{R}^{1,1})$ associated with the global time function $\tau$. Upon this choice we have
$$
\omega_1\simeq \langle\delta \bar\eta,\delta\phi_+ -k\delta\bar\phi_-\rangle\wedge d\sigma\,.
$$
Introducing the new field $\varphi=\phi_+ -k\bar\phi_-$, we can rewrite the last form as
$$
\omega_1\simeq \langle\delta \bar\eta,\delta\varphi \rangle\wedge d\sigma\,.
$$
 One can regard this change of variables as resulting from the formal diffeomorphism $h=e^X$ of the jet space generated by the evolutionary vector field $X$ with
$$
X_0=k\bar\phi_-^i\frac{\partial}{\partial \phi^i_+}\,.
$$
It is clear that $h^\ast(\omega_1)=\omega_1^{_{(1)}}$. Furthermore, applying this diffeomorphism to $\Sigma$ yields
$$
h^\ast(\Sigma)\simeq\Sigma_2\simeq\langle \varphi, [\bar\eta,\bar\eta]\rangle d\sigma\,.
$$
Thus, in terms of the new variables $\varphi$, $\bar\phi$, $\eta$, $\bar\eta$ both the presymplectic form $\omega_1$ and the generating functional of the weak Poisson structure
$$
S=\int_{\tau=c}\Sigma_2
$$
become homogeneous if we assign $\varphi$ with zero polyvector degree. The last fact implies that the derived bracket for the $\mathcal{F}$-Hamiltonian $1$-forms
$$
\{A,B\}=(A,(B,\Sigma_2))
$$
satisfies the standard  Jacobi identity for the Lie bracket. In particular, one can easily see   that the $1$-forms
$$
\varphi(\varepsilon)=\langle\varepsilon,\varphi\rangle d\sigma\,,\qquad \varepsilon\in \Lambda^0(\mathbb{R}^{1,1})\otimes \mathcal{G}\,,
$$
satisfy the commutation relations for the affine Lie algebra $\hat{\mathcal{G}}$ of level $k$:
\begin{equation}\label{JJ}
\{\varphi(\varepsilon_1),\varphi(\varepsilon_2)\}\simeq\varphi([\varepsilon_1,\varepsilon_2]) + k\langle \varepsilon, d\varepsilon \rangle\,.
\end{equation}
One can regard the functional
$$
\varphi[\varepsilon]=\int_{\tau=c}\varphi(\varepsilon)
$$
as the value of the field $\varphi$ at $\tau=c$, smeared with test function $\varepsilon$. Striping  Eq. (\ref{JJ})  of the test functions yields the Poisson brackets of fields at definite space-time points. These brackets can then be extended to more general functionals of fields by the Leibniz rule.

In the framework of Peierls' bracket, the commutation relations (\ref{JJ}) for the Lagrange structure (\ref{OOO}) were first derived in \cite{Sh}.

\vspace{0.2 cm}

\noindent \textbf{Acknowledgements.} I am grateful to Dmitry
Kaparulin and Simon Lyakhovich for critical comments on the first
version of the manuscript.  The work was partially supported by the
RFBR grant 13-02-00551.

\appendix

\section{Jet bundles and the variational bicomplex}\label{A}

In this appendix, we briefly recall some basic elements from the theory of jet bundles and variational bicomplex, which are relevant for our discussion. A more systematic exposition of these concepts can be found in   \cite{Anderson}, \cite{Dickey}, \cite{Olver}, \cite{Saunders}.

The starting point of any field theory is a locally trivial fiber
bundle $\pi : E\rightarrow M$ which base is identified with the
space-time manifold and which sections are called \textit{classical
fields}.  For the sake of simplicity, we restrict ourselves to
fields with values in vector bundles, although the subsequent
discussion could be straightforwardly extended to general smooth
bundles. On the other hand, to accommodate bosonic and fermionic
fields, we alow the fibers of $E$  to be superspaces with a given
number of even and odd dimensions; in so doing, the base $M$ remains
a pure even manifold. The Grassmann parity of a homogeneous object
$A$ will be denoted by $\tilde{A}\in \mathbb{Z}_2=\{0,1\}$.

Associated with a vector bundle $\pi: E\rightarrow M$ is the vector bundle $\pi_k: J^kE \rightarrow M$ of $k$-jets of sections of $E$. By definition, the $k$-jet $j^k_x\phi$
at $x\in M$ is just the equivalence class of the section $\phi\in \Gamma(E)$,  where  two sections are considered to be equivalent if they have the same Taylor development of order $k$ at $x\in M$ in some (and hence any) adapted coordinate chart. It follows from the definition that each section $\phi$ of $E$ induces the section $j^k\phi$ of $J^kE$ by the rule $(j^k\phi)(x)=j^k_x\phi$. The latter is called the $k$-jet prolongation of $\phi$.

If $E|_U\simeq\mathbb{R}^m\times U$ is an adapted coordinate chart with local coordinates $(x^i, \phi^a)$, then $(x^i, \phi^a, \phi^a_i,\ldots, \phi^a_{i_1\cdots i_k})$ are local coordinates in $J^kE$ and the induced section $j^k\phi$ is given in these coordinates by
$$
x\quad \mapsto\quad (x, \phi^a(x), \partial_i\phi^a(x),\ldots, \partial_{i_1}\cdots\partial_{i_k}\phi^a(x))\,.
$$
We use the multi-index notation and the summation convention through the paper. A multi-index $I=i_1i_2\cdots i_n$ represents the corresponding set of symmetric covariant indices. The order of the multi-index is given by $|I|=k$. By definition we set $Ij=jI=i_1i_2\cdots i_kj$. With the multi-index notation we can write the partial derivatives of fields as $\partial_{i_1}\cdots \partial_{i_k}\phi^a=\partial_I \phi^a$  and the set of local coordinates on $J^kE|_U$ as $(x^i,\phi^a_I)$, $|i|\leq k$.

Jet bundles come with natural projection $J^kE\rightarrow J^{k-1}E$ defined by forgetting all the  derivatives of order $k$. One can easily see that this projection gives $J^kE$ the structure of an affine bundle over the base $J^{k-1}E$.  Thus, we have the infinite sequence of surjective submersions
\begin{equation}\label{i-lim}
\xymatrix{\cdots\ar[r]& J^3E\ar[r]&J^2E\ar[r]&J^1E\ar[r]&J^0E\simeq E}\,.
\end{equation}
The infinite order jet bundle $J^\infty E$ is now defined as the inverse  limit over the jets of order $k$:
$$
J^\infty E=\lim_{\longleftarrow}J^k E\,.
$$
Let $\Lambda^\ast(J^kE)$ denote the space of differential forms on $J^kE$.  The sequence of projections (\ref{i-lim}) gives rise to the chain of pullback maps
$$
\xymatrix{\cdots&\ar[l]\Lambda^\ast(J^3E)&\ar[l]\Lambda^\ast (J^2E)&\ar[l]\Lambda^\ast(J^1E)&\ar[l]\Lambda^\ast(J^0E)}\,.
$$
This allows one to define the space of differential forms on $\Lambda(J^\infty E)$ as the direct limit
$$
\Lambda^\ast(J^\infty E)=\lim_{\longrightarrow} \Lambda^\ast(J^kE)\,.
$$
 According to this definition each  differential form on $J^\infty E$ is the pullback of a smooth form on some finite jet bundle $ J^kE$. As usual, the smooth functions on $J^\infty E$ are identified with the $0$-forms. For notational simplicity, we will not distinguish between a form on $J^\infty E$ and its representatives in finite dimensional jet bundles. The exterior differential on $\Lambda^\ast (J^\infty E)$ will be denoted by $D$.

The de Rham complex $(\Lambda^\ast (J^\infty E), D)$ of differential forms  on $J^\infty E$ possesses the differential ideal $\mathcal{C}(J^\infty E)$ of contact forms. By definition, $\alpha\in \mathcal{C}(J^\infty E)$ iff $(j^\infty \phi)^\ast \alpha=0$ for all sections $\phi\in \Gamma(E)$.  The ideal $\mathcal{C}(J^\infty E)$ is known to be generated by the contact $1$-forms, which in local coordinates take the form $\delta \phi^a_I:=D\phi^a_I-\phi^a_{Ij}Dx^j$. Using the contact forms, one can split the exterior differential $D$  into the sum of \textit{horizontal} and \textit{vertical differentials}, namely,  $D=d+\delta$ where
$$
d=dx^j\wedge \left( \frac{\partial}{\partial x^j}+\phi^a_{Ij}\frac{\partial}{\partial \phi^a_{I}}\right)\,,\qquad \delta=\delta \phi_I^a\wedge \frac{\partial_l}{\partial \phi^a_I}\,.
$$
It is easy to see that
$$
d^2=0\,,\qquad \delta^2=0\,,\qquad d\delta+\delta d=0\,.
$$
Any $p$-form of $\Lambda^p(J^\infty E)$ can now be written as a finite sum of homogeneous forms
$$
f dx^{i_1}\wedge \cdots\wedge dx^{i_r}\wedge \delta \phi^{a_1}_{I_1}\wedge\cdots\wedge \delta\phi^{a_s}_{I_s}
$$
of horizontal degree $r$ and vertical degree $s$, with $r+s=p$ and $f$ being a smooth function on $J^\infty E$. The \textit{variational bicomplex} is the double complex $(\Lambda^{\ast,\ast}(J^\infty E), d,\delta)$ of differential forms on $J^\infty E$:
\begin{displaymath}
\xymatrix{
&&\vdots &\vdots&&\vdots\\
&0 \ar[r] & \Lambda^{2,0}(J^\infty E) \ar[u]^\delta \ar[r]^-{d} & \Lambda^{2,1}(J^\infty E) \ar[r]^-{d} \ar[u]^\delta & \ldots \ar[r]^-{d} & \Lambda^{2,n}(J^\infty E) \ar[u]^\delta \\
&0 \ar[r] & \Lambda^{1,0}(J^\infty E) \ar[u]^\delta \ar[r]^-{d} & \Lambda^{1,1}(J^\infty E) \ar[r]^-{d} \ar[u]^\delta & \ldots \ar[r]^-{d} & \Lambda^{1,n}(J^\infty E) \ar[u]^\delta  \\
 0\ar[r]&\mathbb{R}\ar[r] & \Lambda^{0,0}(J^\infty E) \ar[u]^\delta \ar[r]^-{d} & \Lambda^{0,1}(J^\infty E) \ar[r]^-{d} \ar[u]^\delta & \ldots \ar[r]^-{d} & \Lambda^{0,n}(J^\infty E) \ar[u]^\delta \\
}\end{displaymath} Here the base manifold $M$ is assumed to be
connected. The important property of the variational bicomplex is
that all the rows and columns of the diagram above are exact.

It is possible to augment the variational bicomplex from below by the de Rham complex of the base manifold:
\begin{displaymath}
\xymatrix{
 0\ar[r]&\mathbb{R}\ar[r] & \Lambda^{0,0}(J^\infty E)  \ar[r]^-{d} & \Lambda^{0,1}(J^\infty E) \ar[r]^-{d}  & \ldots \ar[r]^-{d} & \Lambda^{0,n}(J^\infty E)  \\
 0\ar[r]&\mathbb{R}\ar[r] & \Lambda^{0}(M) \ar[u]^{\pi_\infty^\ast} \ar[r]^-{d} & \Lambda^{1}(M) \ar[r]^-{d} \ar[u]^{\pi_\infty^\ast} & \ldots \ar[r]^-{d} & \Lambda^{n}(M) \ar[u]^{\pi^\ast_\infty}\\
 & & 0 \ar[u]  & 0  \ar[u] &  & 0 \ar[u]
}\end{displaymath}
The augmented bicomplex is also exact in columns.

As with any bicomplex, one can consider the relative cohomology
of  ``$\delta$ modulo $d$''. It is described by the groups $H^{p,q}(J^\infty E; \delta/d)$ which  are  nothing but the standard cohomology groups of the quotient complex $\widetilde{\Lambda}^{p,q}(J^\infty E)=\Lambda^{p,q}(J^\infty E) /d\Lambda^{p,q-1}(J^\infty E)$ with differential induced by $\delta$. In the main text, we use repeatedly the following statement about the relative $\delta$-cohomology.

\begin{proposition}\cite[Sec.19.3.9]{Dickey}\label{Prop-A}
$$
H^{p,q} (J^\infty E; \delta/d)=0\quad \mbox{for}\quad p>0\quad \mbox{and}\quad H^{0,q}(J^\infty E; \delta/d)\simeq \Lambda^q(M)/d\Lambda^{q-1}(M)\,.
$$
\end{proposition}

The quotient $\delta$-complex $\widetilde{\Lambda}^{p,n}(J^\infty E)={\Lambda}^{p,n}(J^\infty E)/d{\Lambda}^{p,n-1}(J^\infty E)$ provides a natural augmentation of the variational bicomplex from the right:
\begin{displaymath}
\xymatrix{&\vdots&\vdots &\\
  \ar[r]^-{d}&\Lambda^{2,n}(J^\infty E)\ar[u]^\delta  \ar[r]^-{p} & \widetilde{\Lambda}^{n,2}(J^\infty E) \ar[r] \ar[u]^{\delta} & 0 \\
\ar[r]^-{d}& \Lambda^{1,n}(J^\infty E) \ar[u]^{\delta} \ar[r]^-{p} & \widetilde{\Lambda}^{1,n}(J^\infty E) \ar[r] \ar[u]^{\delta} & 0\\
\ar[r]^-{d}& \Lambda^{0,n}(J^\infty E) \ar[u]^{\delta} \ar[r]^-{p} & \widetilde{\Lambda}^{0,n}(J^\infty E) \ar[r] \ar[u]^{\delta} & 0\\
}\end{displaymath}
 $p$ being the canonical projection onto the quotient space. Proposition \ref{Prop-A} ensures that the appended column is exact.   The space $\widetilde{\Lambda}^{0,n}(J^\infty E)$ is usually identified with the space of local functionals of fields. The correspondence between the two spaces is established by the assignment
 $$
 \widetilde{\Lambda}^{0,n}(J^\infty E) \ni [a]\quad \mapsto\quad A[\phi]=\int_M (j^{\infty}\phi)^\ast (a)\,,
 $$
 with $\phi$ being a compactly supported section of $E$.

The space $\Lambda^{1,n}(J^\infty E)$ has a distinguished subspace spanned by the \textit{source forms}. These are given by finite sums of the forms
$$
\alpha\wedge \delta \phi^a\,,
$$
where $\alpha \in \Lambda^{0,n}(J^\infty E)$. Using the exactness of the variational bicomplex one can prove the following

\begin{proposition}[\cite{Dickey}]\label{SF}
For any $(1,n)$-form $\alpha$ there exists a unique source form $\beta$ and a $(1,n-1)$-form $\gamma$ such that
$$
\alpha=\beta+d\gamma\,.
$$
The form $\gamma$ is uniquely determined up to a $d$-exact form. In particular, a nonzero source form can never be $d$-exact.
\end{proposition}

Given $\lambda\in \Lambda^{0,n}(J^\ast E)$, we can apply the proposition above to $\delta \lambda$. We get
$$
\delta \lambda = \delta \phi^a\wedge \frac{\delta \lambda}{\delta \phi^a}+ d\gamma\,.
$$
The coefficients $\delta \lambda/\delta \phi^a$ defining the source form are called the \textit{Euler-Lagrange derivative} of the form $\lambda$. Explicitly,
$$
\frac{\delta \lambda}{\delta \phi^a}=(-\partial)_I\frac{\partial \lambda }{\partial \phi^a_I}\,,
$$
where
\begin{equation}\label{dpat}
(-\partial)_I=(-1)^{|I|}\partial_I\,,\qquad \partial_I=\partial_{i_1}\cdots\partial_{i_k}\,, \qquad \partial_i=\frac{\partial}{\partial x^i}+\phi^a_{Ii}\frac{\partial_l}{\partial\phi^a_{I}}\,.
\end{equation}

Dual to the space of $1$-forms on $J^\infty E$ is the space of vector fields ${\frak X}(J^\infty E)$. In terms of local coordinates, the elements of ${\frak X}(J^\infty E)$ are given by the formal series
\begin{equation}\label{X}
X=X^i\frac{\partial}{\partial x^i} + X_{I}^a\frac{\partial_l}{\partial \phi^a_I}\,,
\end{equation}
where $X^i$ and $X_I^a$ are smooth functions on $J^\infty E$. A vector field $X$ is called \textit{vertical} if $X^i=0$.

The operation $i_X$ of contraction of the vector field (\ref{X}) with a differential form is defined as usual: $i_X$ is a differentiation of the exterior algebra $\Lambda^{\ast}(J^\infty E)$ of form degree $-1$ and the Grassmann parity $\widetilde{X}+1$ which action on the basis 1-forms is given by
$$
i_{X}\delta\phi^a_I=X^a_I\,,\qquad i_{X}dx^j=X^j\,.
$$

The operator of the Lie derivative  along the vector field $X$ is defined by the magic Cartan's formula
\begin{equation}\label{CMF}
L_X=Di_X+(-1)^{\tilde{X}}i_XD\,.
\end{equation}
A vertical vector field $X$ is called \textit{evolutionary} if
$$
i_Xd+(-1)^{\tilde{X}}di_X=0\,.
$$
It follows from the definition that  the vector field (\ref{X}) is  evolutionary iff $X^i=0$ and  $X^a_I=\partial_I (X^a)$, where $\partial_I$ is defined by (\ref{dpat}). Hence, any vertical field of the form $X_0=X^a{\partial}/{\partial \phi^a}$ admits a unique prolongation to an evolutionary vector field. We call $X_0$ the \textit{source vector field} for the evolutionary vector field $X$. (Our nomenclature is not standard; most of the authors prefer to call the vector field $X_0$  evolutionary, rather than its  prolongation $X$.)  Note that the Lie derivative along the evolutionary vector field $X$ can be written as $L_X=i_X \delta+(-1)^{\tilde{X}}\delta i_X$. The Lie algebra of all evolutionary vector fields is denoted  by $\mathfrak{X}_{ev}(J^\infty E)$.

\end{document}